\documentclass[journal]{IEEEtran}
\usepackage[utf8x]{inputenc}
\usepackage{paralist}
\usepackage{enumitem}
\usepackage{lipsum}
\usepackage{framed}

\usepackage{color}
\usepackage{bm} 

\usepackage{amsmath}
\usepackage{amsfonts}
\usepackage{amssymb}
\usepackage{mathrsfs}

\usepackage{stfloats}
\usepackage{graphicx}
\usepackage{epstopdf}
\usepackage[caption=false]{subfig}
\graphicspath{{figures/}}

\makeatother
\makeatletter
\renewcommand{\maketag@@@}[1]{\hbox{\m@th\normalsize\normalfont#1}}%
\makeatother

\usepackage{url}
\usepackage{soul}

\DeclareRobustCommand{\mhl}[1]{%
	\ifmmode\text{\rvp{$#1$}}\else\revisedpart{#1}\fi
}

\newcommand{\rvp}[1]{{\color{blue}{#1}}}

\usepackage{amsthm}
\newtheorem{thm}{Theorem}

\newtheorem{rem}{Remark}
\makeatletter
\renewenvironment{proof}[1][\proofname]{%
	\par\pushQED{\qed}\normalfont%
	\topsep6\p@\@plus6\p@\relax
	\trivlist\item[\hskip\labelsep\bfseries#1\@addpunct{.}]%
	\ignorespaces
}{%
	\popQED\endtrivlist\@endpefalse
}
\makeatother

\usepackage{booktabs}
\usepackage{multicol}
\usepackage{threeparttable}
\usepackage{multirow}
\usepackage{tabularx}
\newcolumntype{L}[1]{>{\raggedright\arraybackslash}p{#1}}
\newcolumntype{C}[1]{>{\centering\arraybackslash}p{#1}}
\newcolumntype{R}[1]{>{\raggedleft\arraybackslash}p{#1}}

\usepackage[level]{datetime}  
\newdateformat{ukdate}{\ordinaldate{\THEDAY}~\monthname[\THEMONTH]~\THEYEAR}

\usepackage{algorithm}
\usepackage{algorithmic}
\usepackage{setspace}

\usepackage{hyperref}
\usepackage[nameinlink,noabbrev]{cleveref}

\crefname{equation}{}{}
\crefname{figure}{Fig.}{Fig.}
\crefname{table}{Table}{Table}
\crefname{lemma}{Lemma}{Lemma}
\crefname{prop}{Proposition}{Proposition}
\crefname{thm}{Theorem}{Theorem}
\crefname{defn}{Definition}{Definition}
\crefname{rem}{Remark}{Remark}
 \usepackage[numbers,sort&compress]{natbib}
\setstcolor{red}

\begin{document}
	
	\title{Network-Assisted Full-Duplex Cell-Free mmWave Networks: Hybrid MIMO Processing and Multi-Agent DRL-Based Power Allocation}

	
\author{~Qingrui~Fan,~Yu~Zhang,~\IEEEmembership{Member,~IEEE,}~Jiamin~Li,~\IEEEmembership{Member,~IEEE,}~Dongming~Wang,~\IEEEmembership{Member,~IEEE,}\\~Hongbiao~Zhang,~and~Xiaohu~You,~\IEEEmembership{Fellow,~IEEE}

    \thanks{\emph{Corresponding author: Jiamin Li (e-mail: jiaminli@seu.edu.cn)}}
	\thanks{This work was supported in part by the National Key R\&D Program of China under Grant 2021YFB2900300, by the National Natural Science Foundation of China (NSFC) under Grant 61971127, by the Southeast University-China Mobile Research Institute Joint Innovation Center, by the Major Key Project of PCL (PCL2021A01-2), and by the open research fund of National Mobile Communications Research Laboratory, Southeast University under Grant 2022D11.}
	\thanks{Q. Fan is with the National Mobile Communications Research Laboratory, Southeast University, Nanjing 210096, China (e-mail: qingruifan@qq.com).}
    \thanks{Y. Zhang is with the School of Electronic and Information Engineering, Nanjing University of Information Science and Technology, Nanjing 210044, China, and also with the National Mobile Communications Research Laboratory, Southeast University, Nanjing, 210096 (email: zhangyu@nuist.edu.cn).}
    \thanks{J. Li, D. Wang and X. You are with the National Mobile Communications Research Laboratory, Southeast University, Nanjing 210096, China, and also with the Purple Mountain Laboratories, Nanjing 211111, China (e-mail: jiaminli, wangdm, xhyu@seu.edu.cn).}
	\thanks{H. Zhang is with China Mobile Research Institute, Beijing, China (email: zhanghongbiao@chinamobile.com).}
	}

	\maketitle
	
	\begin{abstract}
This paper investigates the network-assisted full-duplex (NAFD) cell-free millimeter-wave (mmWave) networks, where the distribution of the transmitting access points (T-APs) and receiving access points (R-APs) across distinct geographical locations mitigates cross-link interference, facilitating the attainment of a truly flexible duplex mode. 
To curtail deployment expenses and power consumption for mmWave band operations, each AP incorporates a hybrid digital-analog structure encompassing precoder/combiner functions.
However, this incorporation introduces processing intricacies within channel estimation and precoding/combining design.
In this paper, we first present a hybrid multiple-input multiple-output (MIMO) processing framework and derive explicit expressions for both uplink and downlink achievable rates. 
Then we formulate a power allocation problem to maximize the weighted bidirectional sum rates.
To tackle this non-convex problem, we develop a collaborative multi-agent deep reinforcement learning (MADRL) algorithm called multi-agent twin delayed deep deterministic policy gradient (MATD3) for NAFD cell-free mmWave networks.
Specifically, given the tightly coupled nature of both uplink and downlink power coefficients in NAFD cell-free mmWave networks, the MATD3 algorithm resolves such coupled conflicts through an interactive learning process between agents and the environment.
Finally, the simulation results validate the effectiveness of the proposed channel estimation methods within our hybrid MIMO processing paradigm, and demonstrate that our MATD3 algorithm outperforms both multi-agent deep deterministic policy gradient (MADDPG) and conventional power allocation strategies.
	\end{abstract}
	

	%
	\IEEEpeerreviewmaketitle

	\section{Introduction}
	%
	%
	%
	%
	\IEEEPARstart
{A}s a key physical layer technology, cell-free massive multiple-input multiple-output (MIMO) system was initially proposed by Ngo \emph{et al.} in \cite{ngo2017cell} and has been extensively studied in recent years. 
In contrast to conventional cell-based architecture, cell-free massive MIMO systems deploy numerous access points (APs) that are distributed randomly across a wide area and connected to a central processing unit (CPU) via fronthaul links to jointly serve a much smaller number of user devices. 
As a result, cell-free massive MIMO systems not only eliminate the concept of cell and the problem of cell-edge user rate in cell-based architecture but also inherit the advantages of channel hardening and good propagation effect in classical centralized MIMO systems \cite{9810259,8676377}.

Furthermore, the investigation of millimeter-wave (mmWave) communication combined with cell-free massive MIMO architecture emerges as the pivotal technology for the advancements in fifth-generation (5G) and beyond-5G (B5G) systems \cite{femenias2019cell,8815888,9786576,9609088}.
The utilization of mmWave communication addresses the scarcity of wireless communication spectrum resources, a challenge that has emerged due to the rapid growth of mobile data rates and the proliferation of terminals required by the Internet of Things (IoT) \cite{9947028, 9915296}.
The benefits of mmWave communication lie in its short wavelength, which enables the installation of numerous antenna elements within a confined space. This capability is supported by contemporary low-power complementary metal oxide semiconductor (CMOS) radio frequency (RF) circuit technology.
Meanwhile, the mmWave carrier frequency enables a broader bandwidth allocation, resulting in elevated data transmission rates~\cite{rappaport2013millimeter, 8464682}. 
Hybrid analog-digital architectures are widely employed to counteract the elevated cost and power consumption associated with conventional full RF chain configurations, and the related operations are denoted as hybrid MIMO processing~\cite{heath2016overview, ni2017near}.
However, mmWave communication suffers from critical drawbacks, such as its significant susceptibility to obstacles and serious path loss, which result in communication outages and unstable quality of service (QoS). 

To satisfy the diverse QoS requisites and accommodate the need for asymmetric data flow within the aforementioned system, researchers have extensively studied co-time co-frequency full-duplex (CCFD) technology in recent years~\cite{sabharwal2014band, nguyen2020spectral}. 
This technology aims to enhance system spectral efficiency by potentially doubling that of half-duplex communication mode through simultaneous uplink and downlink transmission within the same time-frequency resource block. 
However, when the CCFD mode is extended to ultra-dense networks like cell-free massive MIMO systems, the existence of cross-link interference (CLI), including inter-user interference (IUI) and inter-AP interference, seriously damages the system performance. 
Thus, the concept of network-assisted full-duplex (NAFD) was originally introduced to jointly integrate half-duplex, CCFD, and hybrid duplex modes in a unified manner~\cite{wang2019performance}. 
In the NAFD system, each AP is segregated into a transmitting AP (T-AP) and a receiving AP (R-AP), physically partitioned to mitigate the self-interference that can arise within the CCFD framework. 
However, the presence of CLI still limits the spectral efficiency of the NAFD cell-free massive MIMO system~\cite{mohammadi2023network}. 
Hence, how to eliminate CLI is an important issue in the research aimed at enhancing the performance of NAFD cell-free mmWave networks. 

There have been numerous research efforts focused on eliminating CLI in recent years. 
For the elimination of the IUI in the NAFD cell-free massive MIMO system, a user scheduling strategy based on a genetic algorithm is proposed in \cite{wang2019performance} to alleviate the uplink-to-downlink interference. Fukue \emph{et al.} propose a joint resource allocation and beamforming optimization scheme, leveraging location-aided channel estimation to mitigate IUI in the mmWave frequency band~\cite{10048919}. 
To further improve the spectral efficiency of the NAFD system, Xia \emph{et al.} design the downlink sparse beamforming and uplink power control approach to reduce CLI while accommodating limited backhaul capacity in NAFD systems \cite{xia2020joint}. 
Furthermore, Li \emph{et al.} propose an equivalent channel estimation scheme based on beamforming training and an efficient power allocation scheme that only relies on slow-changing large-scale fading information~\cite{li2020network}. Moreover,  \cite{li2023network} and \cite{fan2022maddpg} serve as the foundational basis for this study. 
They tackle the intertwined power allocation problem in NAFD cell-free mmWave massive MIMO systems using convex optimization and multi-agent deep reinforcement learning (MADRL) algorithms, in scenarios where perfect channel state information (CSI) is known, respectively.

Motivated by the aforementioned research, this paper explores a hybrid MIMO processing framework aimed at both uplink and downlink signal transmission, as well as bidirectional power allocation in NAFD cell-free mmWave networks. 
The main contributions of this article are as follows
   \begin{itemize}
       \item We propose a series of hybrid MIMO processing designs to eliminate the effects of CLI in our NAFD cell-free mmWave networks, which include uplink and downlink equivalent channel estimation, inter-AP interference channel estimation, and hybrid digital-analog precoding/combining.

       \item  Within this hybrid MIMO processing framework, we formulate an optimization problem for joint power allocation of uplink and downlink under realistic power constraints to enhance the desired signal while reducing the impact of CLI. The aim is to maximize the bidirectional sum rates and further improve the system's performance.
       
       \item To solve the non-convex coupled power allocation problem, we develop a novel MADRL algorithm with each user acting as an agent interacting with the environment instead of the traditional convex optimization approach which involves high computational overhead in \cite{li2023network}. The multi-agent deep deterministic policy gradient (MADDPG) and multi-agent twin delayed deep deterministic policy gradient (MATD3) algorithms are designed to find the optimal policy to maximize long-term rewards.
       
       \item Numerical simulations are carried out to validate the efficacy of the proposed channel estimation schemes and to assess the convergence of the MATD3 algorithm. In addition, the simulation results demonstrate that our MATD3 scheme outperforms the MADDPG algorithm and other conventional power allocation schemes when solving the complex bidirectional power allocation problem.
   \end{itemize}

The remaining part of this paper is organized as follows. Section II formulates the system model. Section III proposes the hybrid MIMO processing for NAFD cell-free mmWave networks. Section IV analyses the bidirectional spectral efficiency. Section V formulates the power allocation problem and develops the MADRL-based algorithm to handle the non-convex optimization problem. section VI presents and analyzes the simulation results. Section VII concludes this paper.

\textbf{Notations}: $a$, $\mathbf{a}$, $\mathbf{A}$ denotes a scalar, vector, and matrix. $\mathbb{C}^{m\times n}$ denotes the dimension of complex matrices $m\times n$. $|| \mathbf{A} ||_F$, $\operatorname{Tr}(\mathbf{A})$, $\mathbf{A}^{\mathrm{T}}$, $\mathbf{A}^{\mathrm{H}}$ and $\mathbf{A}^{-1}$ denote the Frobenius norm, trace, transpose, conjugate transpose and inverse of matrix $\mathbf{A}$, respectively. $\mathbf{I}_M$ represents an identity matrix whose dimension is $M \times M$. $\otimes$ denotes the Kronecker product calculation. $\mathbb{E}\{ \cdot \}$ denotes the expectation operator and vec($\mathbf{A}$) denotes the vector that stacks all the columns of matrix $\mathbf{A}$.
	

	\section{System Model}
Fig. \ref{system_model} illustrates the architecture of NAFD cell-free mmWave network which comprises $N_T$ T-APs, $N_R$ R-APs, $J$ uplink users, and $K$ downlink users where every T-AP and R-AP is equipped with a uniform linear array (ULA) of $\mathrm{N_{AP}}$ antennas and $\mathrm{N_{RF}}$ RF chains and every user just has one single antenna.
     \begin{figure}[!b]
    \centering
    \includegraphics[width=8.5cm]{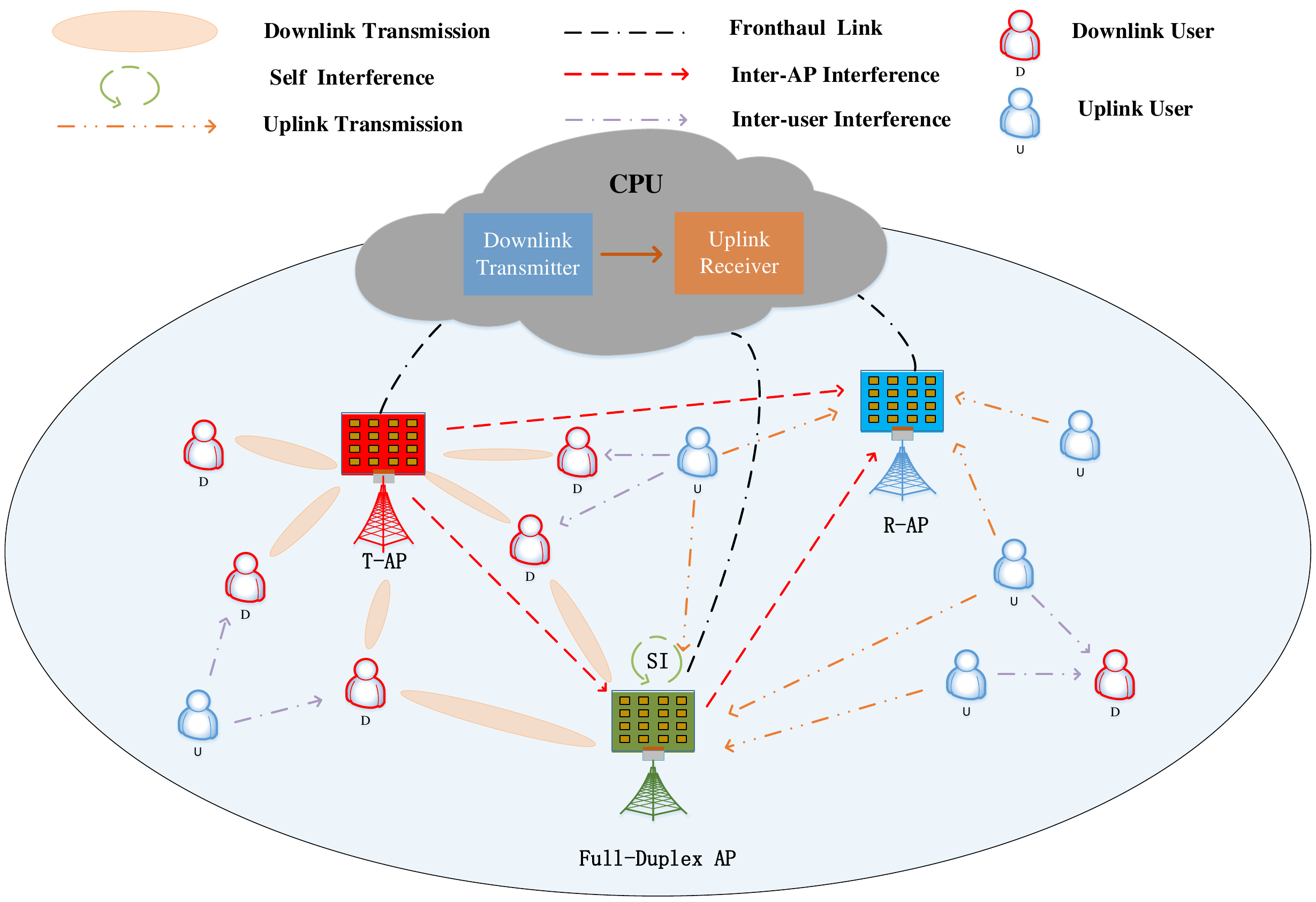}
    \caption{{The framework of NAFD cell-free mmWave networks.}\label{system_model}}
    \end{figure}

Then we model the uplink and downlink mmWave communication channels, inter-AP interference channels, and IUI in the NAFD system, respectively. Firstly, the mmWave downlink channel from the $m$-th T-AP to the $k$-th downlink user with $L$ isotropic paths based on the classical Saleh-Valenzuela model \cite{zhang2020hybrid} is defined as 
\begin{align}
\mathbf{h}_{k, m}^\mathrm{H}=\sum_{l=1}^{L} \alpha_{k, m, l} \mathbf{v}_\mathrm{AP}^\mathrm{H}\left(\theta_{k, m, l}^{\mathrm{AP}}\right)\in\mathbb{C}^{1\times \mathrm{N_{AP}}},
\end{align}
 where $\theta_{k, m, l}^{\mathrm{AP}}$ denotes the angle of departure (AoD) of the $l$-th path between the $m$-th T-AP and the $k$-th downlink user with its range $[-\pi,\pi]$. $\alpha_{k, m, l}$ represents the complex gain which follows the complex gaussian distribution $\mathcal{CN}(0,\beta_{k,m})$ and $\beta_{k,m}$ is the large-scale fading factor. The steering vector of T-AP's ULA can be formulated as $\mathbf{v}_{\mathrm{AP}}\left(\theta_{k, m, l}^{\mathrm{AP}}\right) \triangleq \frac{1}{\sqrt \mathrm{N_{AP}}}\big[1, e^{j \kappa_0  d \sin (\theta_{k, m, l}^{\mathrm{AP}})}, ... , e^{j \kappa_0  d ( \mathrm{N_{AP}} -1) \sin (\theta_{k, m, l}^{\mathrm{AP}})}\big]^{T}$,  where $\kappa_0 = 2\pi / \lambda_c$ with $\lambda_c$ being the mmWave wavelength and $d=\lambda_{c} / 2$ is the distance between adjacent antenna elements. 

 Secondly, the mmWave uplink channel from the $j$-th uplink user to the $z$-th R-AP with $L$ isotropic paths is modeled as
 	\begin{align}\label{}
		\mathbf{g}_{j, z}=\sum_{l=1}^{L} \alpha_{j, z, l} \mathbf{v}_\mathrm{AP}\left(\theta_{j, z, l}^{\mathrm{AP}}\right)\in\mathbb{C}^{\mathrm{N_{AP}}\times 1},    	
	\end{align}
 where $\theta_{j, z, l}^{\mathrm{AP}}$ denotes the angle of arrival (AoA) of the $l$-th path between the $j$-th uplink user and the $z$-th R-AP with its range $[-\pi,\pi]$. $\alpha_{j, z, l}$ denotes the complex gain which satisfies that $\alpha_{j, z, l} \sim \mathcal{CN}(0,\beta_{j,z})$ and $\beta_{j,z}$ is the large-scale fading coefficient. The steering vector of R-AP's ULA $\mathbf{v}_\mathrm{AP}\left(\theta_{j, z, l}^{\mathrm{AP}}\right)$ is designed similar to $\mathbf{v}_{\mathrm{AP}}\left(\theta_{k, m, l}^{\mathrm{AP}}\right)$. 

 Then the inter-AP interference channel matrix between the $m$-th T-AP and the $z$-th R-AP can be modeled as
	\begin{align}\label{}
		\mathbf{H}_{m, z}^\mathrm{AP}=\sum_{l=1}^{L} \alpha_{m, z, l}  \mathbf{v}_\mathrm{AP}\left(\theta_{m, z, l}^{r}\right) & \mathbf{v}_\mathrm{AP}^{\mathrm{H}}\left(\theta_{m, z, l}^{t}\right)\in\mathbb{C}^{ \mathrm{N_{AP}}\times  \mathrm{N_{AP}}},
	\end{align}
where $ \mathbf{v}_\mathrm{AP}\left(\theta_{m, z, l}^{r}\right)$ and $ \mathbf{v}_\mathrm{AP}\left(\theta_{m, z, l}^{t}\right)$ represent the ULA steering vector of R-AP and T-AP with their definitions similar to $\mathbf{v}_{\mathrm{AP}}\left(\theta_{k, m, l}^{\mathrm{AP}}\right)$. $\theta_{m, z, l}^{r}$ and $\theta_{m, z, l}^{t}$ are the AoA and AoD of the $l$-th path whose angle range is $[-\pi,\pi]$. $\alpha_{m, z, l}$ denotes the complex gain of the $l$-th path following $\mathcal{CN}(0,\beta_{m, z})$ with $\beta_{m, z}$ being the large-scale fading coefficient. $\mathbf{R}_{m,z}^\mathrm{AP}$ represents the covariance matrix of $\mathrm{vec}(\mathbf{H}_{m, z}^\mathrm{AP})$ calculated by (\ref{R_m,z}).
	\begin{figure*}[!bp]
	\hrulefill
			\begin{align}\label{R_m,z}
\mathbf{R}_{m,z}^{AP}&= \mathbb{E}\left\{\Vert \mathrm{vec}(\mathbf{H}_{m, z}^{AP})\Vert^2\right\}\nonumber\\
		&=\mathbb{E}\Bigg\{\sum_{l=1}^L \alpha_{m, z, l} \mathrm{vec}\bigg(\mathbf{v}_\mathrm{AP}(\theta_{m, z, l}^{r})\mathbf{v}_\mathrm{AP}^{\mathrm{H}}(\theta_{m, z, l}^{t})\bigg)\nonumber \sum_{l'=1}^L \mathrm{vec}^\mathrm{H}\bigg(\mathbf{v}_\mathrm{AP}\left(\theta_{m, z, l'}^{r}\right)\mathbf{v}_\mathrm{AP}^{\mathrm{H}}\left(\theta_{m, z, l'}^{t}\right)\bigg) \alpha_{m, z, l'}^* \Bigg\}\nonumber\\
		&=\beta_{m,z}\sum_{l=1}^L\mathrm{vec}\bigg(\mathbf{v}_\mathrm{AP}(\theta_{m, z, l}^{r})\mathbf{v}_\mathrm{AP}^{\mathrm{H}}(\theta_{m, z, l}^{t})\bigg)\mathrm{vec}^\mathrm{H}\bigg(\mathbf{v}_\mathrm{AP}\left(\theta_{m, z, l}^{r}\right)\mathbf{v}_\mathrm{AP}^{\mathrm{H}}\left(\theta_{m, z, l}^{t}\right)\bigg).
			\end{align}
     \end{figure*}
 
At last, the IUI coefficient between the $j$-th uplink user and the $k$-th downlink user is modeled as
	\begin{align}\label{}
		t_{k, j} \sim \mathcal{CN}\left(0, \beta_{k, j}\right),
	\end{align}
where $\beta_{k, j}$ is the large-scale fading factor. 

 \section{Hybrid MIMO Processing for NAFD Cell-Free mmWave Networks}
 This section presents a framework of hybrid MIMO processing for NAFD cell-free mmWave networks, including the inter-AP interference channel estimation, the user-AP equivalent channel estimation, and hybrid digital-analog precoding/combining during both channel estimation and data transmission.
 
\subsection{Inter-AP Interference Channel Estimation}
In this subsection, we propose a minimum mean squared error (MMSE) channel estimator for the inter-AP interference to derive the subsequent accurate closed-form rate expressions. Take the channel between the $m$-th T-AP and the $z$-th R-AP for example, the pilot signal received by the $z$-th R-AP is modeled as
	\begin{align}\label{Y_z^AP}
		\mathbf{Y}_{z}^\mathrm{AP}=\sqrt{\rho_{\tau_{AP}}} \mathbf{U}_{z}^{\mathrm{H}} \mathbf{H}_{m, z}^\mathrm{AP} \mathbf{W}_{m} \boldsymbol{\Phi} +\mathbf{N}_\mathrm{AP},
	\end{align}
where $\rho_{\tau_{AP}}$ is the power of pilot training sequences, $\mathbf{U}_{z}\in\mathbb{C}^{ \mathrm{N_{AP}}} \times \mathrm{N_{RF}}$ and $\mathbf{W}_{m}\in\mathbb{C}^{ \mathrm{N_{AP}}} \times \mathrm{N_{RF}}$ represent the analog matrices of the $z$-th R-AP and the $m$-th T-AP respectively in the channel estimation stage. $\boldsymbol{\Phi}\in\mathbb{C}^{ \mathrm{N_{RF}}} \times \tau_{AP}$ is the orthogonal pilot matrix with $\boldsymbol{\Phi}\boldsymbol{\Phi}^\mathrm{H} = \mathbf{I}_{ \mathrm{N_{RF}}}$ and $\mathbf{N}_\mathrm{AP}\in\mathbb{C}^{\mathrm{N_{RF}}} \times \tau_{AP}$ models the additive white Gaussian noise (AWGN) matrix with its elements being independently and identically distributed (i.i.d) and following the distribution $\mathcal{CN}(0,\sigma_{\tau_{AP}}^2)$. After multiplying $\mathbf{Y}_{z}^\mathrm{AP}$ by $\boldsymbol{\Phi}^\mathrm{H}$, \eqref{Y_z^AP} can be reformulated as
	\begin{align}\label{tilde_Y}
		\tilde{\mathbf{Y}}_{z}^\mathrm{AP} &=\mathbf{Y}_{z}^\mathrm{AP}\boldsymbol{\Phi}^\mathrm{H} \nonumber \\ 
  & = \sqrt{\rho_{\tau_{AP}}} \mathbf{U}_{z}^{\mathrm{H}} \mathbf{H}_{m, z}^\mathrm{AP} \mathbf{W}_{m} +\mathbf{N}_\mathrm{AP}\boldsymbol{\Phi}^\mathrm{H},
	\end{align}
	where $\mathbf{N}_\mathrm{AP}\boldsymbol{\Phi}^\mathrm{H}\in\mathbb{C}^{\mathrm{N_{RF}}} \times  \mathrm{N_{RF}}$ also obeys the same distribution as $\mathbf{N}_\mathrm{AP}$ because of the orthogonality of $\boldsymbol{\Phi}$. Using the Kronecker product property $\mathrm{vec}(\mathbf{ABC}) = (\mathbf{C}^\mathrm{T}\otimes \mathbf{A})\mathrm{vec}(\mathbf{B})$, \eqref{tilde_Y} can be rewritten as
	\begin{align}\label{vec_Y}
		\mathrm{vec}(\tilde{\mathbf{Y}}_{z}^\mathrm{AP})= \sqrt{\rho_{\tau_{AP}}} \mathbf{A} \mathrm{vec}(\mathbf{H}_{m, z}^\mathrm{AP}) + \mathrm{vec}(\mathbf{N}),
	\end{align}
	where $\mathbf{A}=\mathbf{W}_{m}^\mathrm{T} \otimes \mathbf{U}_{z}^{\mathrm{H}}$, $\mathbf{N}=\mathbf{N}_\mathrm{AP}\boldsymbol{\Phi}^\mathrm{H}$ and $\mathrm{vec}(\mathbf{N})\sim\mathcal{CN}(\mathbf{0}, \sigma_{\tau_{AP}}^2\mathbf{I}_{ \mathrm{N_{RF}^2}})$.
 Based on the received pilot signal \eqref{vec_Y} and assuming that the AP is aware of the channel statistics, the channel estimate for $\mathrm{vec}(\mathbf{H}_{m, z}^\mathrm{AP})$ can be obtained using MMSE estimation \cite{1597555} as shown in the following Theorem.

\begin{thm}
    The MMSE channel estimate of $\mathrm{vec}(\mathbf{H}_{m, z}^\mathrm{AP})$ is given by
    	\begin{align}\label{}
		\mathrm{vec}(\hat{\mathbf{H}}_{m, z}^\mathrm{AP}) &= \sqrt{\rho_{\tau_{AP}}} \mathbf{R}_{m,z}^\mathrm{AP} \mathbf{A}^\mathrm{H} \\
 &\times (\rho_{\tau_{AP}}\mathbf{A}\mathbf{R}_{m,z}^\mathrm{AP}\mathbf{A}^\mathrm{H} + \sigma_{\tau_{AP}}^2\mathbf{I}_{ \mathrm{N_{RF}^2}})^{-1} \mathrm{vec}(\tilde{\mathbf{Y}}_{z}^\mathrm{AP}), \nonumber
	\end{align} 
where the optimal coupling matrix $\mathbf{A}$ is designed as
	\begin{align}\label{optimal}
		\mathbf{A}=\mathbf{\Sigma}_A\mathbf{U}_R^\mathrm{H},
	\end{align}
where $\mathbf{\Sigma}_A$ is obtained by solving a typical water-filling problem and $\mathbf{U}_R$ is calculated from the eigenvalue decomposition $\mathbf{R}_{m,z}^\mathrm{AP} = \mathbf{U}_R \mathbf{\Lambda}_R \mathbf{U}_R^\mathrm{H}$.
\end{thm}
\begin{proof}
Please refer to Appendix A.
\end{proof}

\begin{rem}
    In NAFD cell-free mmWave networks, the inter-AP interference channel is quasi-static due to the mostly stationary positions of the T-APs and R-APs, while the frequency of the user-AP channel estimation is higher than that of the inter-AP channel estimation due to the mobility of the uplink and downlink users themselves.
\end{rem}

\subsection{User-AP Channel Estimation}
\subsubsection{Downlink Channel Estimation}\label{equivalent_down} 
The specific derivation of the downlink channel estimation from the $m$-th T-AP to the $k$-th downlink user is as follows. The received pilot signal at the $m$-th T-AP can be modeled as
	\begin{align}\label{}
\mathbf{Y}_{m}^{\tau_d}&=\sum_{k=1}^{K}\sqrt{\rho_{\tau_d}} (\mathbf{W}_{m}^\mathrm{RF})^\mathrm{H} \mathbf{h}_{k,m} \boldsymbol{\varphi}_k^\mathrm{T} +\mathbf{N}_{m}^{\tau_d}, \nonumber \\
&=\sum_{k=1}^{K}\sqrt{\rho_{\tau_d}} \mathbf{h}_{k,m}^{eq} \boldsymbol{\varphi}_k^\mathrm{T} +\mathbf{N}_{m}^{\tau_d},
	\end{align}
	where $\mathbf{h}_{k,m}^{eq}\triangleq(\mathbf{W}_{m}^\mathrm{RF})^\mathrm{H} \mathbf{h}_{k,m}$ is defined as the downlink equivalent channel. $\boldsymbol{\varphi}_k \in\mathbb{C}^{\tau_d \times 1}$ is the pilot sequence assigned to the $k$-th downlink user and $\tau_d$ indicates the length of the pilot sequence which is assumed to be $\tau_d \geq K$ so that there is no effect of pilot contamination. $\rho_{\tau_d}$ denotes the total transmit power of the pilot training sequence for all the downlink users and $\mathbf{N}_{m}^{\tau_d}\in\mathbb{C}^{\mathrm{N_{RF}}} \times \tau_d$ is modeled as the AWGN matrix where all entries are i.i.d that follow $\mathcal{CN}(0,\sigma_{\tau_d}^2)$. 
 $\mathbf{W}_{m}^\mathrm{RF}\in\mathbb{C}^{\mathrm{N_{AP}}} \times \mathrm{N_{RF}}$ represents the analog beamforming matrix of the $m$-th T-AP and is designed based on the SLNRmax-SA method \cite{hong2021effect} which satisfies its constant-modulus constraint.

Specifically, it is assumed that the analog precoder of the $m$-th T-AP $\mathbf{W}_{m}^{\mathrm{RF}}$ is solely dependent on the spatial channel covariance matrix $\mathbf{R}_{\mathbf{h}_{k,m}}$ and each AP is considered to have the ability to obtain $\mathbf{R}_{\mathbf{h}_{k,m}}$. The covariance matrix varies very slow over time and it can be easily obtained through spatial channel covariance estimation for hybrid analog-digital MIMO precoding architectures \cite{adhikary2014joint, park2018spatial, femenias2019cell}. In addition, the spatial channel covariance matrix $\mathbf{R}_{\mathbf{h}_{k,m}}$ can be calculated as follow based on the fact that different sub-paths are independent channels
	\begin{align}\label{}
		&\mathbf{R}_{\mathbf{h}_{k,m}}=\mathbb{E}\left\{\mathbf{h}_{k, m}\mathbf{h}_{k, m}^{\mathrm{H}}\right\}\nonumber\\
		&=\mathbb{E}\left\{\sum_{l=1}^L \mathbf{v}_\mathrm{AP}\left(\theta_{k, m, l}^{\mathrm {AP}}\right) \alpha_{k,m,l}^* \sum_{l'=1}^L \alpha_{k,m,l'} \mathbf{v}_\mathrm{AP}^{\mathrm{H}}\left(\theta_{k, m, l'}^{\mathrm{AP}}\right) \right\}\nonumber \\
		&=\beta_{k,m}\sum_{l=1}^L \mathbf{v}_\mathrm{AP}\left(\theta_{k, m, l}^{\mathrm{AP}}\right)\mathbf{v}_\mathrm{AP}^{\mathrm{H}}\left(\theta_{k, m, l}^{\mathrm{AP}}\right).
	\end{align}

Then we calculate the arithmetic average of multiple downlink users in preparation for the design of analog beamforming matrix $\mathbf{W}_{m}^\mathrm{RF}$, which is defined as  $\overline{\mathbf{R}}_{D,m} = \sum_{k=1}^K \mathbf{R}_{\mathbf{h}_{k,m}}/K$. Therefore, the design of $ \mathbf{W}_{m}^\mathrm{RF}$ is shown as follows
	\begin{align}\label{}
		\mathbf{W}_{m}^\mathrm{RF}=\big[e^{j\angle \mathbf{w}_{m,1}^\mathrm{RF}},\cdots,e^{j\angle \mathbf{w}_{m, \mathrm{N_{RF}}}^\mathrm{RF}}\big],
	\end{align}
where $\angle \mathbf{w}$ and $\mathbf{w}_{m,l}^\mathrm{RF}$ represent the phase angles of all the elements of $\mathbf{w}$ and the eigenvector corresponding to the $l$-th largest eigenvalue of $\overline{\mathbf{R}}_{D,m}$, respectively.
 
Since the equivalent channel has a lower dimension than the real channel between transmitting antennas and receiving antennas, estimating the equivalent channel enables the decrease of the estimation overhead and computational complexity. Considering the MMSE estimation principle, the downlink equivalent
channel $\mathbf{h}_{k,m}^{eq}$ can be estimated as
	\begin{align}\label{}
\hat{\mathbf{h}}_{k,m}^{eq}&=\sqrt{\rho_{\tau_d}}\mathbf{R}_{\mathbf{h}_{k,m}^{eq}}\mathbf{Q}_{D,k,m}^{-1}\mathbf{Y}_{m}^{\tau_d}\boldsymbol{\varphi}_k^*,
	\end{align}
	where 
 \begin{align}
 \mathbf{R}_{\mathbf{h}_{k,m}^{eq}} = \mathbb{E}\left\{\mathbf{h}_{k,m}^{eq} (\mathbf{h}_{k,m}^{eq})^\mathrm{H}\right\} = ( \mathbf{W}_{m}^\mathrm{RF})^\mathrm{H} \mathbf{R}_{\mathbf{h}_{k,m}}  \mathbf{W}_{m}^\mathrm{RF},
 \end{align}
 with $\mathbf{Q}_{D,k,m} = \sum_{k'=1}^K \rho_{\tau_d}\mathbf{R}_{\mathbf{h}_{k',m}^{eq}}\big|\boldsymbol{\varphi}_{k'}^\mathrm{T}\boldsymbol{\varphi}_{k}^*\big|^2 + \sigma_{\tau_d}^2\mathbf{I}_{ \mathrm{N_{RF}}}$. Further, the covariance matrices of the estimated channel $\hat{\mathbf{h}}_{k,m}^{eq}{}$ and its error $\tilde{\mathbf{h}}_{k,m}^{eq}$ can be written by
 \begin{subequations}
 \begin{align}
          \mathbf{R}_{\hat{\mathbf{h}}_{k,m}^{eq}}& = \rho_{\tau_d}\mathbf{R}_{\mathbf{h}_{k,m}^{eq}}\mathbf{Q}_{D,k,m}^{-1}\mathbf{R}_{\mathbf{h}_{k,m}^{eq}}^\mathrm{H},\\
     \mathbf{R}_{\tilde{\mathbf{h}}_{k,m}^{eq}} &= \mathbf{R}_{\mathbf{h}_{k,m}^{eq}} - \rho_{\tau_d}\mathbf{R}_{\mathbf{h}_{k,m}^{eq}}\mathbf{Q}_{D,k,m}^{-1}\mathbf{R}_{\mathbf{h}_{k,m}^{eq}}^\mathrm{H}.
 \end{align}
 \end{subequations}

     
\subsubsection{Uplink Channel Estimation}\label{equivalent_up}
Similar to the downlink channel estimation, for uplink channel estimation by all APs through their RF chains, simultaneous transmission of pilot sequences from all users to APs is required during the training phase.
For the channel vector between the $j$-th uplink user and the $z$-th R-AP, the equivalent channel estimation is derived as follows. 
Firstly, the pilot signal received by the $z$-th R-AP can be modeled as
	\begin{align}\label{}
\mathbf{Y}_{z}^{\tau_u}&=\sum_{j=1}^{J}\sqrt{\rho_{\tau_u}} ( \mathbf{U}_{z}^\mathrm{RF})^\mathrm{H} \mathbf{g}_{j, z} \boldsymbol{\varphi}_j^\mathrm{T} +\mathbf{N}_{z}^{\tau_u},\nonumber \\
&=\sum_{j=1}^{J}\sqrt{\rho_{\tau_u}} \mathbf{g}_{j,z}^{eq} \boldsymbol{\varphi}_j^\mathrm{T} +\mathbf{N}_{z}^{\tau_u},
	\end{align}
	where $\mathbf{g}_{j,z}^{eq}\triangleq( \mathbf{U}_{z}^\mathrm{RF})^\mathrm{H} \mathbf{g}_{j,z}$ is defined as the uplink equivalent channel. $\boldsymbol{\varphi}_j \in\mathbb{C}^{\tau_u \times 1}$ is the pilot sequence assigned to the $j$-th uplink user and the length of the pilot sequence $\tau_u \geq J$ is assumed to satisfy that the pilot sequences among uplink users are fully orthogonal with each other, which guarantees no pilot contamination. $\rho_{\tau_u}$ is the total power of the pilot training sequence for the uplink users. $\mathbf{N}_{z}^{\tau_u}\in\mathbb{C}^{ \mathrm{N_{RF}}} \times \tau_u$ is modeled as the AWGN matrix in which each element is i.i.d and obeys the complex Gaussian distribution $\mathcal{CN}(0,\sigma_{\tau_u}^2)$ with its variance $\sigma_{\tau_u}^2$. $\mathbf{U}_{z}^\mathrm{RF}\in\mathbb{C}^{\mathrm{N_{AP}}} \times \mathrm{N_{RF}}$ depicts the analog beamforming matrix of the $z$-th R-AP and is designed as follows
 	\begin{align}\label{}
		\mathbf{U}_{z}^\mathrm{RF}=\big[e^{j\angle \mathbf{u}_{z,1}^\mathrm{RF}},\cdots,e^{j\angle \mathbf{u}_{z, \mathrm{N_{RF}}}^\mathrm{RF}}\big],
	\end{align}
 where $\angle \mathbf{u}$ and $\mathbf{u}_{z,l}^\mathrm{RF}$ represent the phase angles of all the entries of $\mathbf{u}$ and the eigenvector corresponding to the $l$-th largest eigenvalue of $\overline{\mathbf{R}}_{U,z}$, respectively. In addition, $\overline{\mathbf{R}}_{U,z}$ is calculated as $\overline{\mathbf{R}}_{U,z} = \sum_{j=1}^J \mathbf{R}_{\mathbf{g}_{j,z}}/J$, where $\mathbf{R}_{\mathbf{g}_{j,z}}$ is represented by
 	\begin{align}\label{}
		\mathbf{R}_{\mathbf{g}_{j,z}}&=\mathbb{E}\left\{\mathbf{g}_{j, z}\mathbf{g}_{j, z}^{\mathrm{H}}\right\}\nonumber\\
		&=\mathbb{E}\left\{\sum_{l=1}^L \alpha_{j,z,l}\mathbf{v}_\mathrm{AP}\left(\theta_{j, z, l}^{\text {AP}}\right) \sum_{l'=1}^L \mathbf{v}_\mathrm{AP}^{\mathrm{H}}\left(\theta_{j, z, l'}^{\text {AP }}\right) \alpha_{j,z,l'}^*\right\}\nonumber \\
		&=\beta_{j,z}\sum_{l=1}^L \mathbf{v}_\mathrm{AP}\left(\theta_{j, z, l}^{\text {AP }}\right)\mathbf{v}_\mathrm{AP}^{\mathrm{H}}\left(\theta_{j, z, l}^{\text {AP }}\right).
	\end{align}

Then, based on the MMSE principle, the uplink equivalent channel $\mathbf{g}_{j,z}^{eq}$ can be estimated by
	\begin{align}\label{}
\hat{\mathbf{g}}_{j,z}^{eq}&=\sqrt{\rho_{\tau_u}}\mathbf{R}_{\mathbf{g}_{j,z}^{eq}}\mathbf{Q}_{U,j,z}^{-1}\mathbf{Y}_{z}^{\tau_u}\boldsymbol{\varphi}_j^*,
	\end{align}
	where 
 \begin{align}
     \mathbf{R}_{\mathbf{g}_{j,z}^{eq}} = \mathbb{E}\left\{\mathbf{g}_{j,z}^{eq} (\mathbf{g}_{j,z}^{eq})^\mathrm{H}\right\} = (\mathbf{U}_{z}^\mathrm{RF})^\mathrm{H} \mathbf{R}_{\mathbf{g}_{j,z}} \mathbf{U}_{z}^\mathrm{RF},
 \end{align} 
 with $\mathbf{Q}_{U,j,z} = \sum_{j'=1}^J \rho_{\tau_u}\mathbf{R}_{\mathbf{g}_{j',z}^{eq}} \big|\boldsymbol{\varphi}_{j'}^\mathrm{T}\boldsymbol{\varphi}_j^*\big|^2+ \sigma_{\tau_u}^2\mathbf{I}_{\mathrm{N_{RF}}}$. Further, we can also get the following statistical characteristics of the estimated channel $\hat{\mathbf{g}}_{jz}$ and error channel $\tilde{\mathbf{g}}_{jz}$ 
 \begin{subequations}
     \begin{align}
         \mathbf{R}_{\hat{\mathbf{g}}_{j,z}^{eq}} &= \rho_{\tau_u}\mathbf{R}_{\mathbf{g}_{j,z}^{eq}}\mathbf{Q}_{U,j,z}^{-1}\mathbf{R}_{\mathbf{g}_{j,z}^{eq}}^\mathrm{H},\\
           \mathbf{R}_{\tilde{\mathbf{g}}_{j,z}^{eq}} &= \mathbf{R}_{\mathbf{g}_{j,z}^{eq}} - \rho_{\tau_u}\mathbf{R}_{\mathbf{g}_{j,z}^{eq}}\mathbf{Q}_{U,j,z}^{-1}\mathbf{R}_{\mathbf{g}_{j,z}^{eq}}^\mathrm{H}.
     \end{align}
 \end{subequations}
\begin{rem}
    In NAFD cell-free mmWave networks, the analog beamforming matrix designs of T-APs and R-APs will be used for the derivation of the data transmission process due to the equivalent user-AP channel estimation. In order to avoid the design conflict of the analog beamforming matrix, we estimate the inter-AP channel between transmitting antennas and receiving antennas instead of the equivalent inter-AP channel.
\end{rem}

 \subsection{Data Transmission}
    \subsubsection{Downlink Data Transmission}
 The signal received by the $k$-th downlink user admits the following formulation
	\begin{align}\label{yd,k}
		y_{D,k}
		=&\sum_{m=1}^{N_T} \sum_{{i=1}}^{K} (\mathbf{h}_{k, m}^{eq})^\mathrm{H} \mathbf{f}_{m, {i}} \sqrt{\eta_{m, {i}}} s_{{i}}+\sum_{j=1}^{J} t_{k, j} \sqrt{P_{U, j}} s_{U, j}\nonumber\\
		&+{n_{D, k}},
	\end{align}
where $\mathbf{f}_{m, {i}}\in\mathbb{C}^{ \mathrm{N_{RF}}}\times1$ denotes the digital beamforming matrix. $s_{{i}}\sim\mathcal{CN}(0,1)$ and $\eta_{m, {i}}\geq0$ are the data symbol of the ${i}$-th downlink user and the power coefficient. The second term in \eqref{yd,k} models the IUI from other uplink users in which $s_{U,j}$ and $P_{U,j}$ are the data symbol and the transmitted power of the $j$-th uplink user. $n_{D,k}$ is the noise following the distribution $\mathcal{CN}(0,\sigma _{k}^2)$.


\subsubsection{Uplink Data Transmission}
In the uplink, the received signal of the $z$-th R-AP can be modeled as follows
	\begin{align}\label{y_U,z}
		y_{U,z}
        =\sum_{j=1}^J &\mathbf{g}_{j,z}^{eq}\sqrt{P_{U, j}} s_{U, j}+\sum_{m=1}^{N_T} (\mathbf{U}_{z}^\mathrm{RF})^\mathrm{H} \mathbf{H}_{m, z}^{AP} \mathbf{x}_{m}\nonumber \\
           &+( \mathbf{U}_{z}^\mathrm{RF})^\mathrm{H} \mathbf{n}_{U,z},
	\end{align}
where $\mathbf{x}_{m}=\sum_{i=1}^{K} \mathbf{W}_{m}^\mathrm{RF}  \mathbf{f}_{m,{i}} \sqrt{\eta_{m, {i}}} s_{{i}}\in\mathbb{C}^{\mathrm{N_{AP}}\times1}$ is defined as the transmitted signal of the $m$-th T-AP, $\mathbf{W}_{m}^\mathrm{RF}\in\mathbb{C}^{ \mathrm{N_{AP}}\times \mathrm{N_{RF}}}$ denotes the analog beamforming matrix which is designed in Sec. \ref{equivalent_down}, $\mathbf{U}_{z}^\mathrm{RF}\in\mathbb{C}^{ \mathrm{N_{AP}}\times \mathrm{N_{RF}}}$ is the analog beamforming matrix of the $z$-th R-AP which is designed as shown in Sec. \ref{equivalent_up}, $\mathbf{n}_{U,z} \in \mathbb{C}^{\mathrm{N_{AP}}\times 1}$ is the noise vector following $\mathcal{CN}(\bm{0},\sigma_{z}^2\mathbf{I})$. 	


\section{Spectral Efficiency Analysis}\label{closed-form}
In this section, the closed-form sum-rate expressions of uplink and downlink transmission are derived under the proposed hybrid MIMO processing framework of NAFD cell-free mmWave systems.
   
\subsection{Downlink Spectral Efficiency}
 First, based on the definitions of the downlink equivalent channel $\mathbf{h}_{k,m}^{eq}$ and $\mathbf{h}_{k,m}^{eq}=\hat{\mathbf{h}}_{k,m}^{eq}+ \tilde{\mathbf{h}}_{k,m}^{eq}$ in Sec. \ref{equivalent_down}, \eqref{yd,k} can be further transformed into the following form
		\begin{align}\label{yd,k1}
			y_{D,k}
			= & \sum_{m=1}^{N_T} (\hat{\mathbf{h}}_{k,m}^{eq}+\tilde{\mathbf{h}}_{k,m}^{eq})^{\mathrm{H}}\mathbf{f}_{m, k} \sqrt{\eta_{m, k}} s_{k} \nonumber\\
  & +\sum_{i \neq k} \sum_{m=1}^{N_T}\big(\hat{\mathbf{h}}_{k,m}^{eq}+\tilde{\mathbf{h}}_{k,m}^{eq}\big)^{\mathrm{H}}\mathbf{f}_{m, i}\sqrt{\eta_{m, i}} s_{i}\nonumber \\
   &+\sum_{j=1}^{J} t_{k, j} \sqrt{P_{U, j}} s_{U,j}+ n_{D,k}.
		\end{align}
In order to eliminate the interference of the second term in (\ref{yd,k1}), the zero-forcing (ZF) digital precoder is adopted which can be designed by
	\begin{equation}\label{}
		\sum_{m=1}^{N_T} (\hat{\mathbf{h}}_{k,m}^{eq})^{\mathrm{H}} \mathbf{f}_{m, i}=\left\{\begin{array}{l}
			1 \text { if } i=k \\
			0 \text { if } i \neq k
		\end{array}\right..
	\end{equation}
In addition, the power coefficient $\eta_{m, k}$ is forced to only depend on $k$, so the subscript $m$ is removed from $\eta_{m,k}$, which results in $\bm{\eta}=\operatorname{diag}\left(\eta_{1}, \ldots, \eta_{K}\right)$ \cite{kim2021performance}. 
 
 After the interference cancellation, the signal received by the $k$-th downlink user \eqref{yd,k1} becomes
	\begin{align}\label{y_d,k_chaifen}
		\overline{y}_{D,k}=&\underbrace{\sqrt{\eta_{k}} s_{k}}_{\text{Desired Signal}}+\underbrace{n_{D,k}}_{\text{Noise}} \\
         &+ \underbrace{\sum_{i=1}^{K}\sum_{m=1}^{N_T} (\tilde{\mathbf{h}}_{k,m}^{eq})^\mathrm{H} \mathbf{f}_{m,i}\sqrt{\eta_i} s_i}_{\text{Downlink Estimation Error}} + \underbrace{\sum_{j=1}^{J}  t_{k, j} \sqrt{P_{U, j}} s_{U, j}}_{\text{Inter-User Interference}} \nonumber. 
	\end{align}
As in most existing papers \cite{5898372,8247283}, we consider the achievable data rate, where the downlink channel estimation error term in \eqref{y_d,k_chaifen} is treated as Gaussian noise. Additionally, for the sake of reducing the decoding complexity, the IUI term is also regarded as Gaussian noise. In the following, we derive the lower bound for the $k$-th downlink achievable rate in Theorem 2.
\begin{thm}
    The lower bound for the $k$-th downlink achievable rate is derived as \cite{8247283}
    \begin{align}\label{closed-downlink}
		R_{D,k}^{LB} = \log_2\big(1+\frac{\eta_{k}}{\mathcal{I}_{D,k}^{DEE}+ \mathcal{I}_{D,k}^{IUI}+\sigma_k^2}\big),
	\end{align}
	where 
 \begin{align}
\mathcal{I}_{D,k}^{DEE}&=\sum_{m=1}^{N_T}\sum_{i=1}^K\eta_i\operatorname{Tr}(\mathbf{f}_{m,i}\mathbf{f}_{m,i}^\mathrm{H}\mathbf{R}_{\tilde{\mathbf{h}}_{k,m}^{eq}}), \\
\mathcal{I}_{D,k}^{IUI}&=\sum_{j=1}^J\left| t_{k, j} \sqrt{P_{U, j}}\right|^2.
 \end{align}
\end{thm}
\begin{proof}
As with the work in \cite{8247283}, we consider the achievable data rate for the $k$-th downlink user, which is denoted by \eqref{R_Dk}. By using Jensen's inequality, the lower bound of the $k$-th downlink rate can be derived as \eqref{R_DkLB}. 
	\begin{figure*}[!bp]
	\hrulefill
     	\begin{align}\label{R_Dk}
    R_{D,k}&=\mathbb{E}\bigg\{\log_2 \bigg(1+\frac{\eta_k}{|\displaystyle\sum_{i=1}^{K}\sum_{m=1}^{N_T} (\tilde{\mathbf{h}}_{k,m}^{eq})^\mathrm{H} \mathbf{f}_{m,i}\sqrt{\eta_i} s_i|^2+|\sum_{j=1}^{J}  t_{k, j} \sqrt{P_{U, j}} s_{U, j}|^2+\sigma_k^2}\bigg)\bigg\}.\\
    \label{R_DkLB}
        R_{D,k} &\geq \log_2\bigg(1+\frac{\eta_k}{\mathbb{E}\bigg\{|\displaystyle\sum_{i=1}^{K}\sum_{m=1}^{N_T}( \tilde{\mathbf{h}}_{k,m}^{eq})^\mathrm{H} \mathbf{f}_{m,i}\sqrt{\eta_i} s_i|^2\bigg\}+\mathbb{E}\bigg\{|\sum_{j=1}^{J}  t_{k, j} \sqrt{P_{U, j}} s_{U, j}|^2\bigg\}+\sigma_k^2}\bigg), \nonumber\\
        &=\log_2\bigg(1+\frac{\eta_k}{\displaystyle\sum_{m=1}^{N_T}\sum_{i=1}^K\eta_i\operatorname{Tr}(\mathbf{f}_{m,i}\mathbf{f}_{m,i}^\mathrm{H}\mathbf{R}_{\tilde{\mathbf{h}}_{k,m}^{eq}})+\sum_{j=1}^J\left| t_{k, j} \sqrt{P_{U, j}}\right|^2+\sigma_k^2}\bigg), \nonumber\\
        &\triangleq R_{D,k}^{LB}.
			\end{align}
     \end{figure*}
\end{proof}

\subsection{Uplink Spectral Efficiency}
After eliminating the inter-AP interference based on the estimated channel $\hat{\mathbf{H}}_{m, z}^{AP}$ in the digital domain of CPU~\cite{wang2019performance}, the received signal \eqref{y_U,z} can be rewritten as
		\begin{align}
			\tilde{y}_{U,z}=\sum_{j=1}^J &\mathbf{g}_{j,z}^{eq}\sqrt{P_{U, j}} s_{U, j}+\sum_{m=1}^{N_T} (\mathbf{U}_{z}^\mathrm{RF})^\mathrm{H} \mathbf{\tilde{H}}_{m, z}^{AP} \mathbf{x}_{m}\nonumber \\
           &+(\mathbf{U}_{z}^\mathrm{RF})^\mathrm{H} \mathbf{n}_{U,z},
		\end{align}
	where $\mathbf{g}_{j,z}^{eq} = \hat{\mathbf{g}}_{j,z}^{eq} + \tilde{\mathbf{g}}_{j,z}^{eq}$ and $\mathbf{\tilde{H}}_{m, z}^{AP}$ represents the error matrix of the inter-AP channel estimation with its covariance matrix $\mathbf{C}_{m,z}$ shown in (\ref{former}).
 
Similar to the design of digital precoders of T-APs for downlink transmission, the ZF digital precoder is considered at R-AP to eliminate interference from other uplink users. Then after the interference cancellation, the received signal of the $z$-th R-AP is represented as follows
		\begin{align}\label{y_u,z_chafen}
			\overline{y}_{U,z}= & \underbrace{\sqrt{P_{U, j}} s_{U, j}}_{\text{Desired Signal}}+\underbrace{\mathbf{v}_{zj}( \mathbf{U}_{z}^\mathrm{RF})^\mathrm{H}\mathbf{n}_{U,z}}_{\text{Noise}} \nonumber \\
            &+\underbrace{\sum_{m=1}^{N_T} \sum_{i=1}^{K} \mathbf{v}_{zj}( \mathbf{U}_{z}^\mathrm{RF})^\mathrm{H}\mathbf{\widetilde{H}}_{m, z}^{AP} \mathbf{W}_{m}^\mathrm{RF}\mathbf{f}_{m,i}\sqrt{\eta_i}s_i}_{\text{Inter-AP Residual Interference}},
		\end{align}
where the uplink estimation error plus the inter-AP residual interference are collectively referred to as the total estimation error. $\mathbf{v}_{zj} \in \mathbb{C}^{1 \times \mathrm{N_{RF}}}$ denotes the digital receiving vector of the $z$-th R-AP which is designed by
\begin{equation}\label{}
		\mathbf{v}_{zj} \hat{\mathbf{g}}_{j',z}^{eq}=\left\{\begin{array}{l}
			1 \text { if } j'=j \\
			0 \text { if } j' \neq j
		\end{array}\right..
	\end{equation}

 Similar to the downlink spectral efficiency analysis, we consider the achievable data rate, where the uplink channel estimation error and inter-AP residual interference in \eqref{y_u,z_chafen} are both treated as Gaussian noise. In the following, we derive the lower bound for the $j$-th uplink achievable rate in Theorem 3.
\begin{thm}
The lower bound for the $j$-th uplink achievable rate admits the following form
	\begin{equation}\label{closed-uplink}
		\begin{aligned}
            R_{U,j}^{LB}=\log_{2}\big(1+\frac{P_{U,j}}{\mathcal{I}_{U,j}^{TEE}+\mathcal{I}_{U,j}^{Noise}}\big),
		\end{aligned}
	\end{equation}
	where $\mathcal{I}_{U,j}^{TEE}$ and $\mathcal{I}_{U,j}^{Noise}$ are calculated by \eqref{TEE_j} and \eqref{TN_j} at the bottom of next  page.

\end{thm}

\begin{proof}
Similar to the downlink rate derivation, the achievable data rate for the $j$-th uplink user is written as \eqref{R_Uj}. Using Jensen's inequality, the lower bound of the $j$-th uplink rate is derived as \eqref{R_UjLB}.
\begin{figure*}[!bp]
	\hrulefill
     	\begin{align}
      \label{TEE_j} \mathcal{I}_{U,j}^{TEE}
&=\sum_{j'=1}^JP_{U,j'}\mathbf{v}_{zj}\mathbf{R}_{\tilde{\mathbf{g}}_{j'z}}\mathbf{v}_{zj}^\mathrm{H}+\sum_{m=1}^{N_T} \sum_{i=1}^{K}\eta_i\bigg((\mathbf{W}_{m}^\mathrm{RF}\mathbf{f}_{m,i})^\mathrm{T}\otimes \mathbf{v}_{zj}(\mathbf{U}_{z}^\mathrm{RF})^\mathrm{H}\bigg) \mathbf{C}_{m,z} \bigg((\mathbf{W}_{m}^\mathrm{RF}\mathbf{f}_{m,i})^\mathrm{T}\otimes \mathbf{v}_{zj}(\mathbf{U}_{z}^\mathrm{RF})^\mathrm{H}\bigg)^\mathrm{H}, \\
\label{TN_j} \mathcal{I}_{U,j}^{Noise}&=\bigg|\mathbf{v}_{zj}(\mathbf{U}_{z}^\mathrm{RF})^\mathrm{H}\mathbf{n}_{U,z}\bigg|^2=\sigma_z^2\mathbf{v}_{zj}(\mathbf{U}_{z}^\mathrm{RF})^\mathrm{H} \mathbf{U}_{z}^\mathrm{RF}\mathbf{v}_{zj}^\mathrm{H}.\\
\label{R_Uj} R_{U,j}&=\mathbb{E}\bigg\{\log_2 \bigg(1+\frac{P_{U,j}}{|\mathbf{v}_{zj}(\mathbf{U}_{z}^\mathrm{RF})^\mathrm{H}\mathbf{n}_{U,z}|^2+| \displaystyle\sum_{j=1}^J\mathbf{v}_{zj} \tilde{\mathbf{g}}_{j,z}^{eq}\sqrt{P_{U,j}}s_{U, j}|^2+| \displaystyle\sum_{m=1}^{N_T} \sum_{i=1}^{K} \mathbf{v}_{zj}( \mathbf{U}_{z}^\mathrm{RF})^\mathrm{H}\mathbf{\widetilde{H}}_{m, z}^{AP}\mathbf{W}_{m}^\mathrm{RF}\mathbf{f}_{m,i}\sqrt{\eta_i}s_i|^2}\bigg)\bigg\}.\\
    \label{R_UjLB}
    R_{U,j} &\geq \log_2\bigg(1+\frac{P_{U,j}}{\mathbb{E}\bigg\{| \displaystyle\sum_{j=1}^J\mathbf{v}_{zj} \tilde{\mathbf{g}}_{j,z}^{eq}\sqrt{P_{U,j}}s_{U, j}|^2 + |\sum_{m=1}^{N_T} \sum_{i=1}^{K} \mathbf{v}_{zj}(\mathbf{U}_{z}^\mathrm{RF})^\mathrm{H}\mathbf{\widetilde{H}}_{m, z}^{AP}\mathbf{W}_{m}^\mathrm{RF}\mathbf{f}_{m,i}\sqrt{\eta_i}s_i|^2\bigg\}+|\mathbf{v}_{zj}(\mathbf{U}_{z}^\mathrm{RF})^\mathrm{H}\mathbf{n}_{U,z}|^2}\bigg),\nonumber\\
    &=\log_{2}\big(1+\frac{P_{U,j}}{\mathcal{I}_{U,j}^{TEE}+\mathcal{I}_{U,j}^{Noise}}\big),\nonumber\\
    &\triangleq R_{U,j}^{LB}.
			\end{align}
     \end{figure*}
     
\end{proof}

	\section{MADRL-based Bidirectional Power Allocation}
 \subsection{Problem Formulation}
In this subsection, we introduce a power allocation optimization problem to enhance the desired signal while reducing the inter-AP residual interference and IUI, thereby improving the bidirectional spectral efficiency. Specifically, the objective of the proposed problem is to maximize the weighted uplink and downlink sum rate based on the bidirectional closed-form expressions (\ref{closed-downlink}) and (\ref{closed-uplink}) derived in Sec. \ref{closed-form}, which needs to meet the transmit power limitations of T-APs and uplink users. Thus, the specific optimization problem can be formulated as
    \begin{subequations}
			\begin{alignat}{2}
			&\max _{\eta_{k},P_{U,j}} \quad& &\omega _D \sum_{k=1}^{K} R_{D,k}^{LB}+ \omega_U \sum_{j=1}^{J} R_{U,j}^{LB}\label{mubiao} \\
			& \;\; \text { s.t. }
			& \quad & P_{D,m}\leq P_D ,  \quad \forall m, \label{yueshu1} \\
			&&& P_{U,j}\leq P_U, \quad  \forall j, \label{yueshu3}
		\end{alignat}
	\end{subequations}
where $\eta_{k}$ and $P_{U,j}$ represent the bidirectional optimization variables, (\ref{yueshu1}) and (\ref{yueshu3}) denote the power constraints for the T-APs and uplink users, respectively. We consider the case of $\omega_D+\omega _U = 1$ to ensure the fairness of communication, where $\omega_D$ and $\omega _U$ denote the weights of downlink and uplink sum rate respectively. The transmitted power of the $m$-th T-AP $P_{D,m}$ is defined as 
\begin{align}\label{downlink power limit}
P_{D,m}(\bm{\eta}) &= \mathbb{E}[\Vert \mathbf{x}_m \Vert^2]=\operatorname{Tr}\left( \mathbf{W}_{m}^\mathrm{RF} \mathbf{F}_{m} \boldsymbol{\eta} \mathbf{F}_{m}^{\mathrm{H}} (\mathbf{W}_{m}^\mathrm{RF})^{\mathrm{H}}\right),
\end{align}
which $\mathbf{F}_m = \big[ \mathbf{f}_{m,1}, \cdots, \mathbf{f}_{m,K}\big]$.

    \subsection{The MATD3 Algorithm: An Improved Approach Based On The MADDPG}
    \begin{figure*}[!tp]
    \centering
    \includegraphics[width=0.85\textwidth]{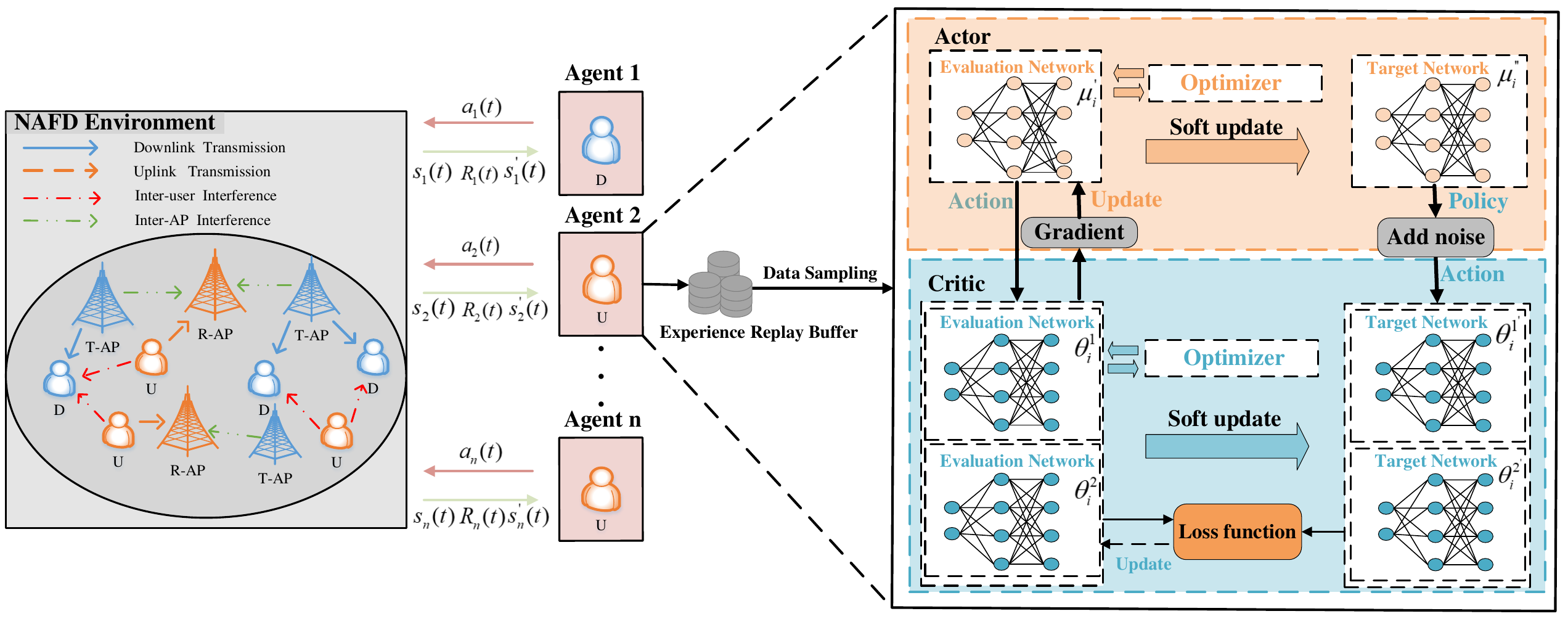}
    \caption{{The MATD3 framework for the power allocation of NAFD cell-free mmWave networks.}\label{structure_MATD3}}
    \end{figure*}
Although some traditional convex optimization approaches can effectively deal with the power allocation problem and obtain the optimal solution which can be verified in our previous work \cite{li2023network} that an algorithm based on successive convex approximation is proposed to solve the joint optimization problem of power allocation and fronthaul compressed noises. However, these methods also have the limitation that the computational complexity escalates with the expansion of the system configuration, encompassing factors such as the quantity of APs, users, and antennas~\cite{lei2021maddpg}.

Based on the above considerations, we use the DRL framework as an alternative algorithm to implement power allocation schemes with low computational complexity. 
Moreover, conventional Q-learning and Deep Q-Network (DQN) algorithms necessitate discretization of continuous state and action spaces, potentially leading to imprecise outcomes and suboptimal system performance.
Meanwhile, increased discretization accuracy can result in a dimensionality catastrophe, leading to substantial computational overhead~\cite{fredj2022distributed}.

Therefore, a centralized training with a distributed execution algorithm named MADDPG \cite{lowe2017multi} is proposed to solve the power allocation problem with continuous action space and state space as the multi-agent version of DDPG algorithm, which is essentially an actor-critic-based approach. In an actor-critic network structure, each agent's actor network is responsible for the local policy evaluation and a centralized critic network is responsible for the $Q$-value evaluation.

However, overestimation bias is a common problem in deep reinforcement learning algorithms because maximizing the state-action Q-value function in the critic network to learn the optimal strategy will lead to unstable and extremely optimistic learning results. Unfortunately, since the MADDPG algorithm maximizes the estimation of the target critic network which is used to learn the optimal strategy, the MADDPG suffers from such overestimation bias problem \cite{lowe2017multi}.

Therefore, the MATD3 is proposed to avoid the overestimation bias of Q-value function and to provide a more stable training process as the multi-agent domain of TD3 algorithm, where two critic networks and two target critic networks are adopted in each agent \cite{ackermann2019reducing}.

It can be seen from Fig. \ref{structure_MATD3} which shows the MATD3 framework for the power allocation of NAFD cell-free mmWave networks that the MATD3 also adopts the centralized training and distributed execution approach. However, different from single critic network of each agent in the MADDPG, the $i$-th agent in the MATD3 algorithm utilizes two evaluation critic networks with parameters $\theta_i^1$ and $\theta_i^2$ and two target critic networks with delayed parameters $\theta_i^{1^{'}}$ and $\theta_i^{2{'}}$ which generates two evaluated Q-values and two target Q-values respectively to reduce the influence of overestimation bias caused by $max$ function in Bellman equation. Then the target Q-value of each agent depends on the minimum of the two target critic networks' Q-values, which can be calculated as
    \begin{align}\label{matd3_q}
        y_i =r_{i}+\gamma \min_{j=1,2} Q_{\theta_i^{j^{'}}}^{\bm{\pi}}\left(\mathbf{s}_j^{\prime},\tilde{a}_j \right),
    \end{align}
where $\mathbf{s}_j^{\prime}$ and $\tilde{a}_j$ indicate the next state and the action generated by the target actor network of the $i$-th agent, respectively. $Q_{\theta_i^{j^{'}}}^{\bm{\pi}}$ denotes the Q-value function of the $i$-th agent. In order to avoid narrow-peak overfitting of Q-value function, random noise is considered to be added to the actions performed by the target actor network in the training process, which can make Q-value estimation smoother \cite{zhao2022multi}. Thus, the action $\tilde{\bm{a}}_j$ in (\ref{matd3_q}) is represented by 
    \begin{align}
\tilde{a}_j=\pi_{\mu_j^{\prime}}\left(s_j^{\prime}\right)+\text{clip}(\mathcal{N}(0,\hat{\sigma}_a^2), -1, 1),
    \end{align}
where $\hat{\sigma}_a$ is the standard deviation of the additive noise.

Moreover, in the training phase, the evaluation critic networks and the target critic networks of each agent have a global perspective, in which the evaluation critic networks of the $i$-th agent are responsible for obtaining the state and action information of other agents. Then all agents use the global information $\bm{s}(t)$ and joint actions $\bm{a}(t)=\big\{a_1(t),\cdots, a_N(t)\big\}$ to obtain their own Q-values and each agent adjusts its own local actor network's strategy to realize the globally optimal policy $\bm{\pi}=\big\{\pi_1,\cdots,\pi_N \big\}$ based on other agents' policies. In order to make the training process more stable, the agent stores the observed information $(\mathbf{s},\mathbf{s'},a_1,\cdots,a_N,r_1,\cdots,r_N)$ in the experience replay buffer $\mathcal{D}$ and update their own evaluation actor network parameters $\bm{\mu}^{\prime}=\big\{\mu_1^{\prime}, \mu_2^{\prime}, \cdots, \mu_N^{\prime}\big\}$ through the policy gradient which is written as
    \begin{align}\label{matd3_actor}
\nabla_{\mu_{i}^{\prime}} J\left(\mu_{i}^{\prime}\right)=&\mathbb{E}_{\mathbf{s}, a_i}\bigg[\nabla_{\mu_i^{\prime}} \pi_{\mu_{i}^{\prime}}\left(s_i\right) \nonumber\\
         &\times\nabla_{a_{i}} Q_{i}^{\bm{\pi}}\left(\mathbf{s}, a_{1}, \ldots, a_N\right)\big|_{a_{i}=\pi_{\mu_{_i}^{\prime}}\left(s_i\right)}\bigg].
    \end{align}

Then the parameters of two evaluation critic networks are updated by the loss function as follows
    \begin{align}\label{loss_eva_criric}
            \mathcal{L}\left(\theta_{i}^j\right) &=\mathbb{E}_{\mathbf{s}, a, r, \mathbf{s}^{\prime}}\left[\left(Q_{\theta_i^{j}}^{\bm{\pi}}\left(\boldsymbol{s}_j, \boldsymbol{a}_j\right)-y_i\right)^{2}\right], \quad j=1,2,
    \end{align}
where $y_i$ is defined in (\ref{matd3_q}). And each agent updates its own three evaluation network parameters based on the following formulations
\begin{align}
\mu_i^{\prime} & \leftarrow \mu_i^{\prime}-\lambda \nabla_{\mu_i^{\prime}} J\left(\mu_i^{\prime}\right), \label{actor_eva}\\
\theta_i^j & \leftarrow \theta_i^j-\lambda \nabla_{\theta_i^j} \mathcal{L}\left(\theta_i^j\right), \quad j=1,2,\label{critic_eva}
\end{align}
where $\lambda$ denotes the learning rate. In addition, the parameters of each agent's evaluation actor network are updated more slowly than those of the evaluation critic networks to ensure that the error of time difference is minimized. We assume that the evaluation actor network is updated every $d$ training step.
At the same time, according to the equation (\ref{actor_eva}) and (\ref{critic_eva}), we update the parameters of the three target networks of each agent, which are defined as follows
    \begin{align}
    & \mu_i^{\prime\prime}=\epsilon \mu_i^{\prime}+(1-\epsilon) \mu_i^{\prime\prime}, \label{target_actor}\\
    & \theta_i^{j^{\prime}}=\epsilon \theta_i^j+(1-\epsilon) \theta_i^{j^{\prime}}, \quad j=1,2,\label{target_critic}
    \end{align}
    where $\epsilon$ denotes the updating rate.

In the distributed execution phase, the $i$-th agent executes its actor policy based on its own local state observations $s_i(t)$ through the trained evaluation actor network. Since agents do not communicate with each other during the execution phase, the communication cost will be extremely reduced which contributes that the MATD3 algorithm can be naturally extended to NAFD cell-free mmWave networks. The specific MATD3 scheme for the power allocation optimization problem is summarized in \textbf{Algorithm \ref{alg1}}. 

        \begin{algorithm}[ht]
		\renewcommand{\algorithmicrequire}{\textbf{Input:}}
		\renewcommand{\algorithmicensure}{\textbf{Output:}}
		\caption{The MATD3 Algorithm for Power Allocation Problem}
		\label{alg1}
		\begin{algorithmic}[1]
			\STATE Initialize each agent's actor networks with parameters $\mu_i^{\prime}$ and $\mu_i^{\prime\prime}$, each agent's critic networks with parameters $\{\theta_i^{j}\}_{j=1,2}$ and $\{\theta_i^{j^{'}}\}_{j=1,2}$, respectively.
			\STATE Initialize each agent's experience replay buffer $\mathcal{D}$.
			\FOR{episode $= 1, 2, \cdots, N_e$}
			\FOR{$t = 1, 2, \cdots, t_{max}$}
			\STATE Based on the current state $\bm{s}(t)$, the uplink and downlink agents perform actions $\bm{a}(t)$ which consist of $P_{U,j}$ and $\eta_k$ respectively. 
			\STATE All the agents obtain the immediate reward $R_k(t)$ and $R_j(t)$ and the next state $\bm{s}^{'}(t)$ after the interaction with the environment.
			\STATE Store the experience tuple $\{ \textbf{s}(t), \textbf{s}^{'}(t), \textbf{a}(t), \mathcal{R}(t)\}$ in $\mathcal{D}$.
			\FOR{$i=1,\cdots,N$}
			\STATE Sample a mini-batch $D$ of experience in $\mathcal{D}$ randomly.
			\STATE Update parameters $\{\theta_i^{j}\}_{j=1,2}$ of the $i$-th agent's evaluation critic networks by minimizing loss function based on the equation (\ref{loss_eva_criric}).
			\IF{$t$ mod $d$} 
            \STATE Update parameter $\mu_i^{\prime}$ of the $i$-th agent's evaluation actor network defined as (\ref{actor_eva}).
            \STATE Update parameters $\mu_i^{\prime\prime}$ and $\{\theta_i^{j^{'}}\}_{j=1,2}$ of the $i$-th agent's three target networks with the formulation (\ref{target_actor}) and (\ref{target_critic}).
            \ENDIF 
			\ENDFOR
			\ENDFOR
			\ENDFOR
		\end{algorithmic}
	\end{algorithm}
    
    \subsection{Design of the MATD3 Algorithm}
    In this section, to deploy the MATD3 algorithm to solve the power allocation problem proposed in this paper, we regard the uplink and downlink users as agents, and design several important elements in the MATD3 algorithm, including the state information observed by agents, the actions of agents and the rewards they achieve \cite{fan2022maddpg}.
    
    \begin{itemize}
    \item \emph{\textbf{Action}}: The actions of each agent depend on the optimization variables in the power allocation problem. Specifically, for the $j$-th uplink user, its action is the transmitting power coefficient $P_{U,j}$, while the $k$-th downlink user's corresponding action is power coefficient $\eta_k$.
        
    \item \emph{\textbf{Observation}}: The observed value of each agent is the estimated equivalent CSI. Specifically, for the $j$-th uplink user agent, it observes the estimated equivalent CSI from all uplink users to all R-APs, i.e. $\hat{\mathbf{g}}_{j,z}$ for all $j$ and $z$ and interference information between uplink users and downlink users, i.e. $t_{k,j}$ for all $k$ and $j$. For the $k$-th downlink user agent, it observes the estimated equivalent CSI from all T-APs to all downlink users, i.e. $\hat{\mathbf{h}}_{km}$ for all $k$ and $m$ and IUI information, i.e. $t_{k,j}$ for all $k$ and $j$.

    \item \emph{\textbf{Reward}}: At the $t$-th training moment, each agent interacts with the environment to get its immediate reward $R_{D,k}^{LB}(t)$ and $R_{U,j}^{LB}(t)$, which feeds back to the agent and guides the subsequent decision-making action. Since the uplink user's power range is set to meet the power constraint (\ref{yueshu3}) initially, we design the immediate reward at the $t$-th training moment of bidirectional agents as follows
    \begin{align}
        &R_{D,k}^{LB}(t) = \omega_D R_{D,k}^{LB} + \beta*{\rm clip}\left(P_D-P_{D,m},-1,1\right), \label{r_k}\\
        &R_{U,j}^{LB}(t) = \omega_U R_{U,j}^{LB},
    \end{align}
    where $\beta$ is the appropriate punishment factor verified by multiple simulations and the second term of the above equation (\ref{r_k}) means that a positive feedback is applied if the power constraint is satisfied, otherwise the opposite. Therefore, the total immediate reward at the $t$-th time slot is
    \begin{align}
        R(t) = \sum_{k=1}^K R_{D,k}^{LB}(t) + \sum_{j=1}^J R_{U,j}^{LB}(t).
    \end{align}
    \end{itemize}
	
	\section{Simulation Results}
In this section, numerical experiments are conducted to evaluate the effectiveness of our proposed mmWave MIMO channel estimation schemes and the performance comparison of MADRL algorithms.

\subsection{Parameter Settings}\label{setting}
Firstly, we consider an NAFD cell-free mmWave network comprised of $N_T = 6$ T-APs, $N_R = 6$ R-APs, $J = 4$ uplink users and $K = 4$ downlink users randomly deployed in a circular area with radius $R=60$ m, where all the APs are equipped with $\mathrm{N_{AP}} = 6$ antennas and $ \mathrm{N_{RF}} = 3$ RF chains, and each user just has a single antenna. Unlike CCFD systems, T-APs and R-APs in NAFD systems are geographically separated in this area and the protection distance between users and APs is $5$ m. In addition, we consider mmWave communication over the $28$ GHz carrier frequency and the $100$ MHz channel bandwidth with a total noise power $\sigma_k^2 = \sigma_z^2=-85$ dBm, the maximum transmit power of T-APs $P_D = 30$ dBm and the transmit power limitation of users $P_U = 27$ dBm. The path loss generated by large-scale fading at mmWave frequency is modeled as
	\begin{align}
		\mathrm{PL}(d)[\mathrm{dB}]=\mathrm{PL}(d_0)+10\xi \log _{10}(\frac{d}{d_0})+X_{\zeta},
	\end{align}
where $\mathrm{PL}(d_0)=20\log_{10}(\frac{4\pi d_0}{\lambda_{c}})$ and $d_0$ represents the free space path loss and a reference distance of $1$ m, respectively. $\xi$ is the path loss index set to $2.92$ and $X_{\zeta}$ is the logarithm shadow fading coefficient which follows a normal distribution $\mathcal{N}(0,\sigma_{log}^2)$ with its variance $8.7^2$. 
 
The performance of our proposed mmWave MIMO channel estimation scheme is measured by normalized mean square error (NMSE) and defined as
	\begin{align}
		 \mathrm{NMSE}=\mathbb{E}\left\{\|\hat{\mathbf{H}}_{m,z}-\mathbf{H}_{m,z}\|_\mathrm{F}^2 /\|\mathbf{H}_{m,z}\|_\mathrm{F}^2\right\}.
	\end{align}
 In the context of the proposed MADRL-based power allocation, the actor and critic networks of each agent are realized by using fully connected neural networks comprising two hidden layers. Each of these hidden layers is composed of $64$ nodes and is subsequently activated by the Tanh function. The dimension of the final layer of the actor network aligns with that of the agent's action. The important hyperparameters associated with the MATD3 scheme are concisely presented in \textbf{Table~\ref{table_simulation}}.
  
	\begin{table}
	\caption{Hyperparameters of the MATD3 algorithm}
	\label{table_simulation}
	\setlength{\tabcolsep}{0.9mm}
	\centering
	\begin{tabular}{l l}
		\toprule[1.5pt]
		\textbf{Hyperparameter} & \textbf{Value} \\
		\midrule
		Number of agents \qquad \qquad  \qquad\qquad\qquad\qquad\qquad & 8  \\
		Mini-batch size $D$ & 1024\\
		Learning rate $\lambda$   & 0.0005 \\
		Discount factor $\gamma$ & 0.95 \\
		Number of units &64\\
		Policy update rate &2\\
		Max training step $t_{max}$ in Algorithm \ref{alg1}  &50\\ 
		Critic-action noise  &0.2\\
		Action noise  &0.5\\
            Optimizer  & Adam\\
		\bottomrule[1.5pt]
	\end{tabular}
    \end{table}

\subsection{Evaluation of Channel Estimation}\label{channel_estimation}

     \begin{figure}[!h]
    \centering
    \includegraphics[width=7cm]{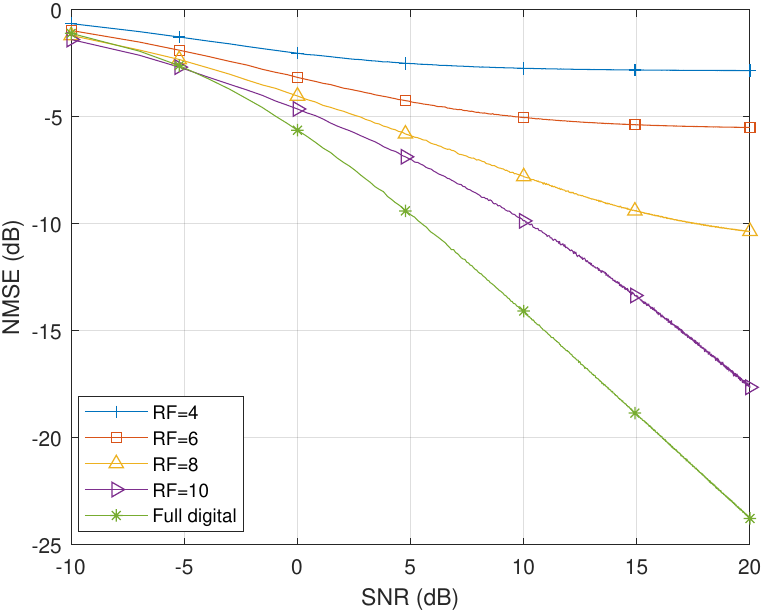}
    \caption{{NMSE for different RF chains and \textbf{full digital} case with $\mathrm{N_{AP}}=32$.}\label{NMSE1}}
    \end{figure}
     \begin{figure}[!h]
    \centering
    \includegraphics[width=7cm]{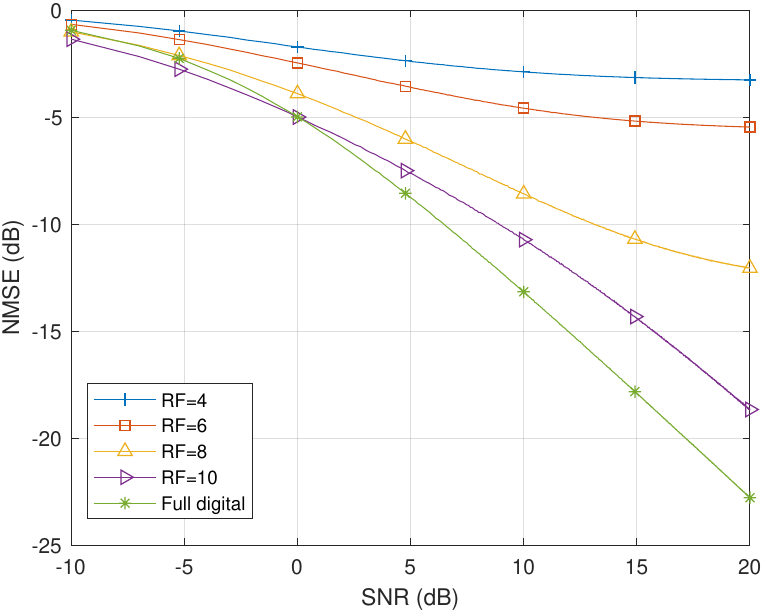}
    \caption{{NMSE for different RF chains and \textbf{full digital} case with $\mathrm{N_{AP}}=128$.}\label{NMSE2}}
    \end{figure}
Fig. \ref{NMSE1} and Fig. \ref{NMSE2} show the NMSE versus signal-to-noise ratio (SNR) for the proposed channel estimation algorithm under various RF chain configurations and the \textbf{full digital} scheme. The scenarios consider $32$ and $128$ antennas, respectively. 
The so-called \textbf{full digital} scheme aligns the number of antennas with that of  RF chains, a configuration prevalent in sub-6G communication precoding design. 
Analysis of Figs. \ref{NMSE1} and \ref{NMSE2} reveal that with an increase in RF chains, the NMSE values decrease for identical SNR levels. 
This suggests that increasing the number of RF chains improves channel estimation accuracy, although it consistently exhibits lower performance compared to the full-digital counterpart.
Specifically, as shown in Fig. \ref{NMSE1}, when the SNR is $-10$ dB, the performance of the MMSE-based channel estimation approach is poor under different numbers of RF chains and \textbf{full digital} scheme. 
However, when the SNR increases to $20$ dB, the value of NMSE under \textbf{full digital} case reaches nearly $-25$ dB. 
By contrast, the NMSE of the estimated channel at $10$ RF chains is about $-17.5$ dB, while the NMSE value at $4$ RF chains is only $-3$ dB, which indicates a extremely poor estimation performance. Meanwhile, in the practical applications of mmWave systems, Fig. \ref{NMSE1} can be referenced to select an appropriate number of RF chains in the hybrid precoder to achieve near-optimal performance at a lower cost.


Additionally, as the number of antennas increases, the NMSE values depicted in Fig. \ref{NMSE1} and Fig. \ref{NMSE2} indicate comparable performance in channel estimation between $128$ and $32$ antennas.
This observation further demonstrates the efficacy of the proposed channel estimation algorithm.
The proposed algorithm capitalizes on the mmWave MIMO channel linking T-AP and R-AP antennas. 
Notably, the performance remains robust as the dimensionality of the estimated channel increases.
A specific comparison between Fig. \ref{NMSE1} and Fig. \ref{NMSE2} highlights the similarity in NMSE values for scenarios with the same number of RF chains but varying antenna numbers. 
For instance, for an AP equipped with $10$ RF chains, the NMSE value is approximately $-18$ dB for both $32$ and $128$ antennas at an SNR of $20$ dB.
In contrast, the \textbf{full digital} case yields a similar value of approximately $-23$ dB.

\subsection{Comparison of MADRL Algorithms}
    \begin{figure}[!h]
    \centering
    \includegraphics[width=7cm]{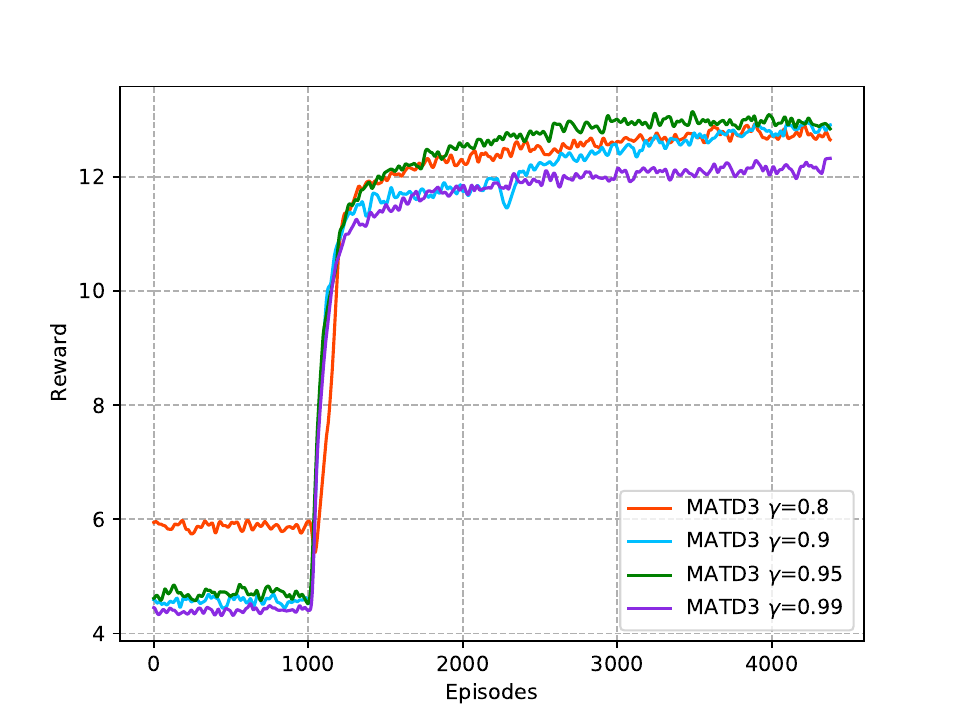}
    \caption{{The average reward with different $\gamma$ based on the MATD3.}\label{MATD3_different_gamma}}
    \end{figure}
Fig. \ref{MATD3_different_gamma} tests the average reward convergence for different discount factors $\gamma$. As can be seen from Fig. \ref{MATD3_different_gamma}, small discount factors, such as $\gamma=0.8$ and $\gamma=0.9$, could lead the MATD3 algorithm to converge to local optimal solutions.
Conversely, a high discount factor, e.g., $\gamma=0.99$,  overly emphasizes future rewards, potentially disrupting the strategy's responsiveness to immediate rewards. 
In summary, the MATD3 algorithm achieves the best convergence performance with $\gamma=0.95$. 
    
    \begin{figure}[!h]
    \centering
    \includegraphics[width=7cm]{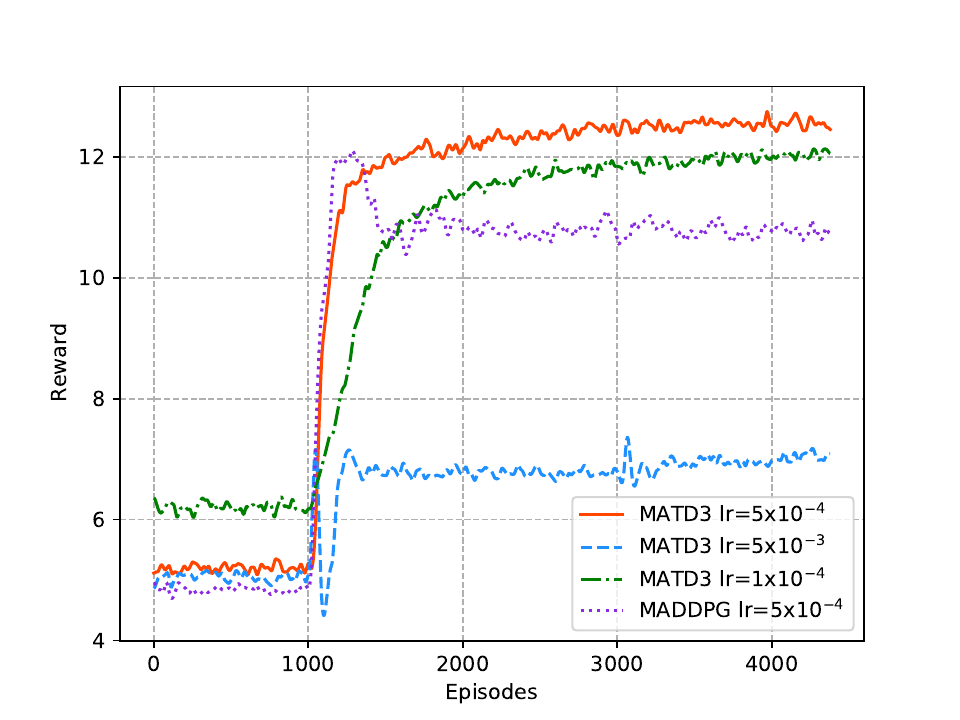}
    \caption{{The average reward with different learning rate based on the MATD3 and MADDPG.}\label{MATD3_different_lr_rate}}
    \end{figure}
    Fig. \ref{MATD3_different_lr_rate} illustrates the average reward curve for different learning rates based on the MATD3 and MADDPG algorithms. 
    Overall, the scheme with the learning rate $0.0005$ achieves the best average reward values compared with the learning rates $0.005$ and $0.0001$. 
    Specifically, when the learning rate is $0.005$, the learning curve finally converges to $7$, which is about half lower than the convergence values for the learning rates $0.0005$. 
    This is due to the fact that an excessive exploration rate may cause the MATD3 algorithm to fall into a local optimal solution, resulting in relatively low convergence values. 
    In addition, when the learning rate is $0.0001$, the speed of convergence becomes slower compared to the learning rate of $0.0005$ because of the huge difference between the evaluation network and the target network. Consequently, taking into account the training process speed and the average reward value during convergence, we identify $0.0005$ as the optimal learning rate for the MATD3 scheme in the following experiments.
    
    It can also be seen from Fig.~\ref{MATD3_different_lr_rate} that the MADDPG training curve first increases to $12$, then falls, and finally converges around $10.9$. 
    In contrast, MATD3 displays a more consistent learning performance and attains a higher convergence value of $12.8$, signifying an enhancement of the overall average reward by approximately $1.9$. 
    This is due to the fact that the MATD3 utilizes dual-centralized critic networks to overcome the overestimation error, utilizes target policy smoothing to reduce variance, and utilizes delayed policy updates to ensure a stable learning process. 
    In addition, it's worth noting that unlike supervised learning characterized by well-defined labels, training curves in MADRL algorithms exhibit specific oscillations.

    \begin{figure}[!h]
    \centering
    \includegraphics[width=7cm]{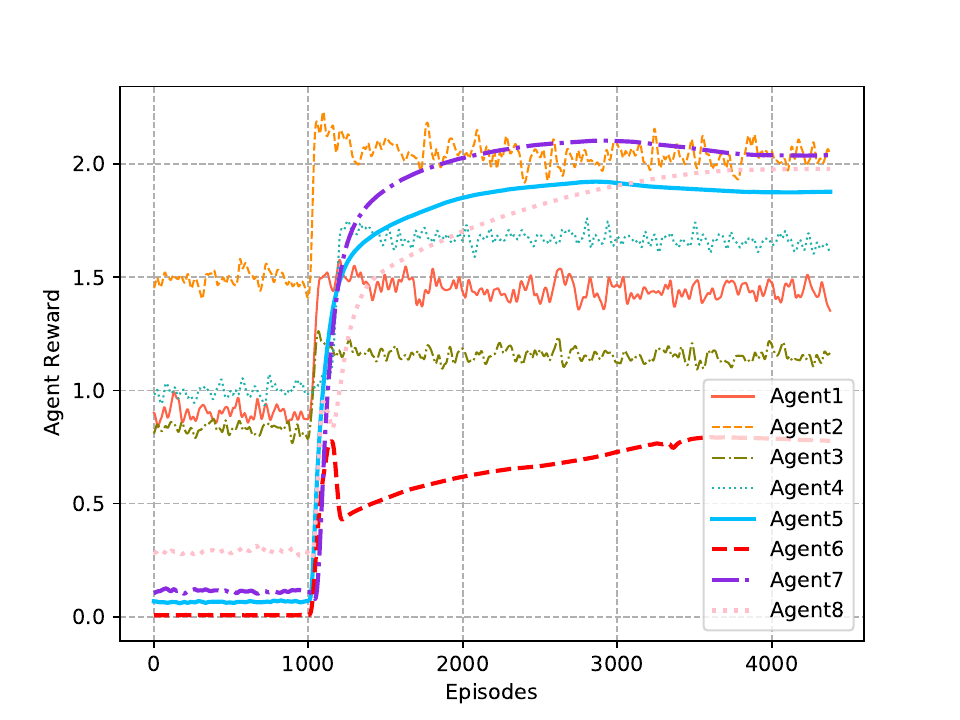}
    \caption{{The training curves of every agent average reward with the learning rate $0.0005$.}\label{agent_reward}}
    \end{figure}
Fig. \ref{agent_reward} represents the training curves of each agent with the learning rate $0.0005$. 
In terms of the optimization problem to be solved, there are eight agents, where agent $1$ $\sim$ $4$ corresponds to four uplink users, and agent $5$ $\sim$ $8$ corresponds to four downlink users. 
This figure demonstrates a consistent pattern: as training episodes advance, reward values for all agents gradually increase and ultimately converge.
Due to individual agents operating within distinct environments, their actions yield diverse rewards, culminating in various levels of convergence. 
For example, the final convergence value of agent $6$ is around $0.78$, which is much smaller than the convergence value around $2.04$ of agent $2$. 
Nevertheless, the agents collaborate, working towards the overarching goal of maximizing the cumulative long-term return.

    \begin{figure}[!h]
    \centering
    \includegraphics[width=7cm]{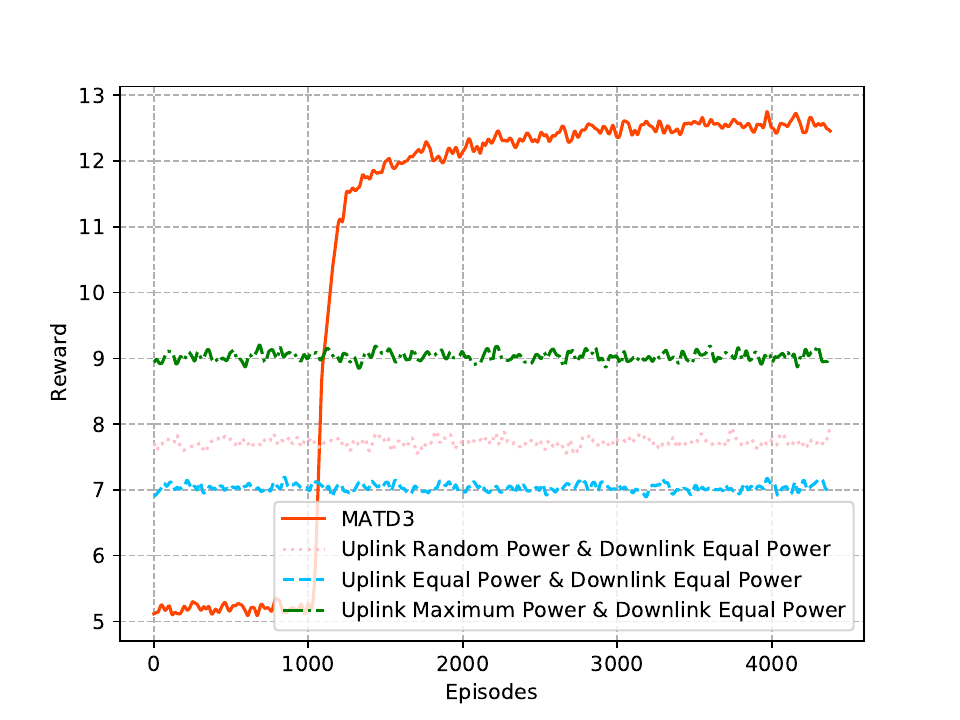}
    \caption{{The average reward for the MATD3 and different power allocation schemes.}\label{MATD3_different_Power}}
    \end{figure}
    
   Fig. \ref{MATD3_different_Power} compares the performance of the MATD3 algorithm with conventional power allocation schemes: a) uplink random power $\&$ downlink equal power allocation, b) uplink equal power $\&$ downlink equal power allocation, and c) uplink maximum power $\&$ downlink equal power allocation.
   It is worth noting that due to the interplay of power coefficients in the downlink power constraint described by \eqref{downlink power limit}, we consider downlink equal power, where $\bm{\eta}$ is set to $P_D/\mathrm{max}(\operatorname{Tr}( \mathbf{W}_{m}^\mathrm{RF} \mathbf{F}_{m} \mathbf{F}_{m}^{\mathrm{H}} ( \mathbf{W}_{m}^\mathrm{RF})^{\mathrm{H}}))$. 
   Overall, in terms of convergence performance, the MATD3 $\textgreater$ scheme c $\textgreater$ scheme a $\textgreater$ scheme b. Specifically, the final convergence reward of the MATD3 is improved by $3.8$ compared to scheme c, by $5$ compared to scheme a, and by $5.8$ compared to scheme b, which is consistent with results caused by different uplink power schemes.   
    \begin{figure}[!h]
    \centering
    \includegraphics[width=7cm]{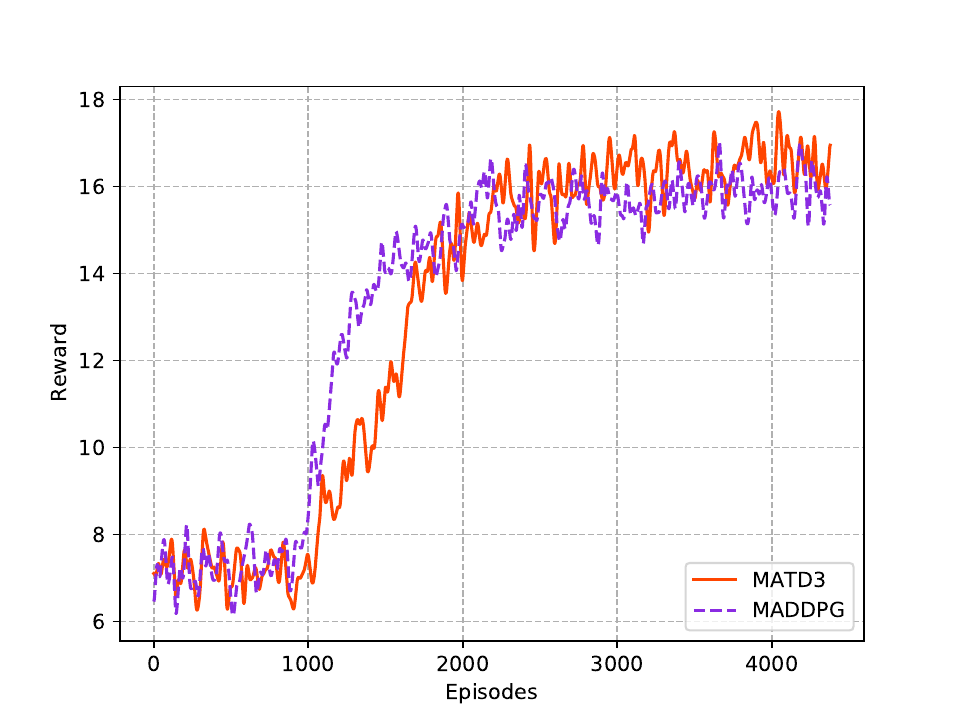}
    \caption{{The average reward for the MATD3 and MADDPG in dynamic environment.}\label{MATD3_MADDPG_dynamic}}
    \end{figure}

    Fig. \ref{MATD3_MADDPG_dynamic} demonstrates the learning performance of the MATD3 and MADDPG in dynamic environments, where the users' positions undergo periodic changes.
    As can be seen from Fig. \ref{MATD3_MADDPG_dynamic}, the two algorithms end up oscillating and converging within a certain interval.    
    Remarkably, the MATD3 algorithm achieves marginally superior rewards than the MADDPG algorithm, albeit at the expense of a little reduced convergence speed.
	\section{Conclusion}
In this paper, we analyze the NAFD cell-free mmWave networks within a hybrid MIMO processing paradigm. 
In this context,, we first design inter-AP interference channel estimation, user-AP channel estimation, and hybrid digital and analog precoding/combining. 
Building upon these hybrid processing designs, we derive closed-form expressions for bidirectional sum rates during data transmission.
We then present a joint power allocation optimization problem aimed at maximizing the weighted bidirectional sum rate while satisfying the constraints of power limitations. 
To tackle the inherent non-convex power allocation challenge in NAFD cell-free mmWave networks, we propose a collaborative MADRL framework named the MATD3 algorithm. 
The proposed algorithm efficiently obtains the optimal strategy, bypassing the computational complexity and overhead associated with convex optimization. 
Through the simulation results, we verify the effectiveness of the proposed channel estimation schemes and the convergence of the MATD3 algorithm. 
Furthermore, we demonstrate that the MATD3 algorithm outperforms the MADDPG optimization scheme and other conventional power allocation methods.

\begin{appendices}
\numberwithin{equation}{section}
\section{Appendix A Proof of Theorem 1}
According to the MMSE criterion~\cite{1597555}, we can obtain the channel estimate $\mathrm{vec}(\hat{\mathbf{H}}_{m, z}^{AP})$ and its error covariance matrix $\mathbf{C}_{m,z}$ written as
\begin{subequations}
     	\begin{align}
\mathbf{C}_{m,z} &= \mathbb{E} \big\{\Vert\mathrm{vec}(\mathbf{H}_{m, z}^{AP})-\mathrm{vec}(\hat{\mathbf{H}}_{m, z}^{AP})\Vert^2\big\} \nonumber \\
		&= \mathbf{R}_{m,z}^{AP}-\rho_{\tau_{AP}}\mathbf{R}_{m,z}^{AP}\mathbf{A}^\mathrm{H}\nonumber \\
  &\times (\rho_{\tau_{AP}}\mathbf{A}\mathbf{R}_{m,z}^{AP}\mathbf{A}^\mathrm{H} + \sigma_{\tau_{AP}}^2\mathbf{I}_{ \mathrm{N_{RF}^2}})^{-1}\mathbf{A}\mathbf{R}_{m,z}^{AP}\label{former}\\
		&=((\mathbf{R}_{m,z}^{AP})^{-1} + \rho_{\tau_{AP}} \sigma_{\tau_{AP}}^{-2}\mathbf{A}^\mathrm{H}\mathbf{A})^{-1}\label{latter}.
			\end{align}
 \end{subequations}

  Therefore, the channel estimation error between the $m$-th T-AP and the $z$-th R-AP is defined by
	\begin{align}
		\text{MSE}_{m,z} \triangleq \operatorname{Tr}\big[((\mathbf{R}_{m,z}^{AP})^{-1} + \rho_{\tau_{AP}} \sigma_{\tau_{AP}}^{-2}\mathbf{A}^\mathrm{H}\mathbf{A})^{-1}\big].
	\end{align}

Thus, we propose an objective function that minimizes the channel estimation error $\text{MSE}_{m,z}$ to design the optimal coupling matrix $\mathbf{A}$ which is formulated as follows
	\begin{subequations}
		\begin{alignat}{2}
			&\min _{\mathbf{A}} \quad& & \text{MSE}_{m,z}\label{mubiao1} \\
			& \;\; \text { s.t. }
			& \quad & \mathbf{A}\in\mathbb{C}^{ \mathrm{N_{RF}^2 }\times  \mathrm{N_{AP}^2}},  \quad  \label{yueshu1_epi} \\
			&&& \operatorname{Tr}(\mathbf{A}\mathbf{A}^\mathrm{H})\leq  \mathrm{N_{RF}^2},  \quad  \label{yueshu2_epi}
		\end{alignat}
	\end{subequations}
where \eqref{yueshu1_epi} and \eqref{yueshu2_epi} denote the dimensional constraint and maximum power limitation of the coupling matrix $\mathbf{A}$, respectively. The singular value decomposition and eigenvalue decomposition for $\mathbf{A}$ and $\mathbf{R}_{m,z}^{AP}$ can be carried out as $\mathbf{A} = \mathbf{U}_A \mathbf{\Sigma}_A \mathbf{V}_A^\mathrm{H}$ and $\mathbf{R}_{m,z}^{AP} = \mathbf{U}_R \mathbf{\Lambda}_R \mathbf{U}_R^\mathrm{H}$, respectively. According to the solution of similar optimization problem in \cite{fozooni2018hybrid}, if the eigenvectors of $\mathbf{A}^\mathrm{H}\mathbf{A}$ and $(\mathbf{R}_{m,z}^{AP})^{-1}$ are the same, the effect of minimizing the channel estimation error can be achieved. Further, based on the assumption that $\mathbf{V}_A = \mathbf{U}_R$, the optimization problem (\ref{mubiao1}) can be rewritten as follows
	\begin{subequations}
		\begin{alignat}{2}
			&\min _{\mathbf{\Sigma}_A} \quad& &\operatorname{Tr}\bigg(\big(\mathbf{\Lambda}_R^{-1} + \rho_{\tau_{AP}} \sigma_{\tau_{AP}}^{-2}\mathbf{\Sigma}_A^\mathrm{H}\mathbf{\Sigma}_A\big)^{-1}\bigg)\label{mubiao2} \\
			& \;\; \text { s.t. }
			& \quad & \mathbf{\Sigma}_A\in\mathbb{C}^{ \mathrm{N_{RF}^2} \times  \mathrm{N_{AP}^2}},  \quad  \label{} \\
			&&& \operatorname{Tr}(\mathbf{\Sigma}_A^\mathrm{H}\mathbf{\Sigma}_A)\leq \mathrm{N_{RF}^2}. \quad  \label{}
		\end{alignat}
	\end{subequations}

Then through the following definitions that $x_i$ represents the $i$-th largest eigenvalue of $\mathbf{A}^\mathrm{H}\mathbf{A}$ and $y_i$ denotes the $i$-th smallest eigenvalue of $(\mathbf{R}_{m,z}^{AP})^{-1}$, the above optimization problem can be further simplified into a typical water-filling problem which can be reformulated by
	\begin{subequations}
		\begin{alignat}{2}
			&\min _{x_i} \quad& &\sum_{i=1}^{ \mathrm{N_{AP}^2}} \bigg(\frac{1}{y_i + \rho_{\tau_{AP}}\sigma_{\tau_{AP}}^{-2} x_i}\bigg)\label{mubiao3} \\
			& \;\; \text { s.t. }
			& \quad & \sum_{i=1}^{ \mathrm{N_{AP}^2}} x_i\leq \mathrm{N_{RF}^2},  \quad  \label{} \\
			&&& x_i\geq 0,  \quad \mathrm{for}\quad i=1,\cdots, \mathrm{N_{RF}^2}, \label{}\\
			&&& x_i= 0,  \quad \mathrm{for}\quad i= \mathrm{N_{RF}^2}+1,\cdots, \mathrm{N_{AP}^2}. \label{}
		\end{alignat}
	\end{subequations}
 
According to Karush-Kuhn-Tucker (KKT) conditions~\cite{boyd2004convex}, we can get the following expressions \eqref{KKT} on the top of next page, where $v_1$ is a number that needs to satisfy the equation $\sum_{i=1}^{\mathrm{N_{AP}^2} } x_i= \mathrm{N_{RF}^2}$.
 	\begin{figure*}[htbp]
	\hrulefill
	    \begin{align}\label{KKT}
    &\sqrt{\frac{\rho_{\tau_{AP}}\sigma_{\tau_{AP}}^{-2}}{v_1}} = y_i + \rho_{\tau_{AP}}\sigma_{\tau_{AP}}^{-2} x_i\nonumber=\frac{\rho_{\tau_{AP}}\sigma_{\tau_{AP}}^{-2} \mathrm{N_{RF}^2} + \displaystyle\sum_{i=1}^{\mathrm{{N_{RF}^{2}}^{'}}} y_i}{\mathrm{N_{RF}^{2}}^{'}},\\
    &x_i = \frac{\sigma_{\tau_{AP}}^2}{\rho_{\tau_{AP}}}\bigg(\frac{\rho_{\tau_{AP}}\sigma_{\tau_{AP}}^{-2}\mathrm{N_{RF}^2} + \displaystyle \sum_{i=1}^{\mathrm{N_{RF}^{2}}^{'}} y_i}{\mathrm{N_{RF}^{2}}^{'}}- y_i \bigg), \nonumber \qquad \mathrm{for} \quad i=1,\cdots,\mathrm{N_{RF}^{2}}^{'},\nonumber \\
    &x_i = 0,\qquad \mathrm{for} \quad i={\mathrm{N_{RF}^{2}}}^{'},\cdots, \mathrm{N_{AP}^2}.
    	\end{align}
     \end{figure*}
Therefore, the optimal coupling matrix $\mathbf{A}$ is finally designed as
\setcounter{equation}{6}
	\begin{align}\label{}
\mathbf{A}=\mathbf{U}_A\mathbf{\Sigma}_A\mathbf{U}_R^\mathrm{H},
	\end{align}
where the main diagonal elements of $\mathbf{\Sigma}_A$ are $[\sqrt{x_1},\sqrt{x_2},\cdots,\sqrt{x_{\mathrm{N_{RF}^{2}}^{'}}},0,\cdots,0]^\mathrm{T}$ with its dimension $ \mathrm{N_{RF}^2}\times \mathrm{N_{AP}^2}$. In addition, since $\mathbf{U}_A$ does not participate in the optimization process of the above problem, we assume that $\mathbf{U}_A=\mathbf{I}_{\mathrm{N_{RF}^2}}$, and $\mathbf{A}$ is further designed as
	\begin{align}\label{optimal}
		\mathbf{A}=\mathbf{\Sigma}_A\mathbf{U}_R^\mathrm{H}.
	\end{align}
	
Although we can obtain the design of the optimal coupling matrix $\mathbf{A}$ through formula \eqref{optimal}, we could not get the coupling matrix information between the transmitter and the receiver in the actual system, so we need to solve the practical RF matrix at both the transmitter and receiver to approximate $\mathbf{A}$ maximally.

Then we propose the following optimization problem to solve the optimal coupling matrix $\mathbf{W}_m^\mathrm{T}$ and $\mathbf{U}_{z}^{\mathrm{H}}$
    	\begin{subequations}
		\begin{alignat}{2}
			&\mathop{\arg \min}_{\mathbf{W}_m^\mathrm{T},\mathbf{U}_z^\mathrm{H}} \quad& &\Vert \mathbf{A}-\mathbf{W}_{m}^\mathrm{T} \otimes \mathbf{U}_{z}^{\mathrm{H}} \Vert_F\\
			& \;\; \text { s.t. }
			& \quad &\mathbf{U}_{z}\in\mathbb{C}^{\mathrm{N_{AP}} \times \mathrm{N_{RF}}},  \quad \label{ch5_youhua1} \\
			&&& \mathbf{W}_{m}\in\mathbb{C}^{ \mathrm{N_{AP}} \times  \mathrm{N_{RF}}}, \label{ch5_youhua2}
		\end{alignat}
	\end{subequations}
where \eqref{ch5_youhua1} and \eqref{ch5_youhua2} represent the dimensional constraints of RF matrices at the receiver and the transmitter respectively. After permuting the shape of the matrix $\mathbf{A}$, the above optimization problem can be rewritten as 
    	\begin{subequations}
		\begin{alignat}{2}
			&\mathop{\arg \min} _{\mathbf{W}_m^\mathrm{T},\mathbf{U}_z^\mathrm{H}} \quad& &\Vert \tilde{\mathbf{A}}-\mathrm{vec}(\mathbf{W}_{m}^\mathrm{T}) \mathrm{vec}(\mathbf{U}_{z}^{\mathrm{H}})^\mathrm{T} \Vert_F\\
			& \;\; \text { s.t. }
			& \quad &\mathbf{U}_{z}\in\mathbb{C}^{ \mathrm{N_{AP}} \times \mathrm{N_{RF}}},  \quad \label{ch5_youhua11} \\
			&&& \mathbf{W}_{m}\in\mathbb{C}^{\mathrm{N_{AP}} \times \mathrm{N_{RF}}}, \label{ch5_youhua22}
		\end{alignat}
	\end{subequations}
where the design of $\mathbf{\tilde{A}}$ is referred to the approach shown in \cite{van1993approximation}. Further, the above optimization problem can be solved according to the Eckhart-Young theorem. To be more specific, let $\mathbf{\tilde{A}} = \sum_r \sigma_r \mathbf{u}_r \mathbf{v}_r^\mathrm{T}$ denotes the singular value decomposition of matrix $\mathbf{\tilde{A}}$ in which $\sigma_r$ is the $r$-th largest singular value, $\mathbf{u}_r$ and $\mathbf{v}_r$ are the corresponding left singular vector and right singular vector, respectively. Eventually, the optimal solutions of $\mathbf{W}_m^\mathrm{T}$ and $\mathbf{U}_{z}^{\mathrm{H}}$ for the best approximation to $\mathbf{A}$ are designed as follows
    	\begin{subequations}
		\begin{alignat}{2}
			&\mathrm{vec}(\mathbf{W}_{m}^\mathrm{T}) = \sqrt{\sigma_1}\mathbf{u}_1,\\
			&\mathrm{vec}(\mathbf{U}_{z}^{\mathrm{H}}) = \sqrt{\sigma_1}\mathbf{v}_1.
		\end{alignat}
	\end{subequations}



\end{appendices}
	
	%

	

\begin{thebibliography}{10}
\providecommand{\url}[1]{#1}
\csname url@samestyle\endcsname
\providecommand{\newblock}{\relax}
\providecommand{\bibinfo}[2]{#2}
\providecommand{\BIBentrySTDinterwordspacing}{\spaceskip=0pt\relax}
\providecommand{\BIBentryALTinterwordstretchfactor}{4}
\providecommand{\BIBentryALTinterwordspacing}{\spaceskip=\fontdimen2\font plus
\BIBentryALTinterwordstretchfactor\fontdimen3\font minus
  \fontdimen4\font\relax}
\providecommand{\BIBforeignlanguage}[2]{{%
\expandafter\ifx\csname l@#1\endcsname\relax
\typeout{** WARNING: IEEEtran.bst: No hyphenation pattern has been}%
\typeout{** loaded for the language `#1'. Using the pattern for}%
\typeout{** the default language instead.}%
\else
\language=\csname l@#1\endcsname
\fi
#2}}
\providecommand{\BIBdecl}{\relax}
\BIBdecl

\bibitem{ngo2017cell}
H.~Q. Ngo, A.~Ashikhmin, H.~Yang, E.~G. Larsson, and T.~L. Marzetta,
  ``Cell-free massive {MIMO} versus small cells,'' \emph{IEEE Trans. Wirel.
  Commun.}, vol.~16, no.~3, pp. 1834--1850, Mar. 2017.

\bibitem{9810259}
J.~Kassam, D.~Castanheira, A.~Silva, R.~Dinis, and A.~Gameiro, ``Distributed
  hybrid equalization for cooperative millimeter-wave cell-free massive
  {MIMO},'' \emph{IEEE Trans. Commun.}, vol.~70, no.~8, pp. 5300--5316, Aug.
  2022.

\bibitem{8676377}
M.~Alonzo, S.~Buzzi, A.~Zappone, and C.~D’Elia, ``Energy-efficient power
  control in cell-free and user-centric massive {MIMO} at millimeter wave,''
  \emph{IEEE Trans. Green Commun. Netw.}, vol.~3, no.~3, pp. 651--663, Sept.
  2019.

\bibitem{femenias2019cell}
G.~Femenias and F.~Riera-Palou, ``Cell-free millimeter-wave massive {MIMO}
  systems with limited fronthaul capacity,'' \emph{{IEEE} Access}, vol.~7, pp.
  44\,596--44\,612, Mar. 2019.

\bibitem{8815888}
Y.~Jin, J.~Zhang, S.~Jin, and B.~Ai, ``Channel estimation for cell-free mmwave
  massive {MIMO} through deep learning,'' \emph{IEEE Trans. Veh. Technol.},
  vol.~68, no.~10, pp. 10\,325--10\,329, Oct. 2019.

\bibitem{9786576}
U.~Demirhan and A.~Alkhateeb, ``Enabling cell-free massive {MIMO} systems with
  wireless millimeter wave fronthaul,'' \emph{IEEE Trans. Wirel. Commun.},
  vol.~21, no.~11, pp. 9482--9496, Nov. 2022.

\bibitem{9609088}
J.~Wang, B.~Wang, J.~Fang, and H.~Li, ``Millimeter wave cell-free massive
  {MIMO} systems: Joint beamforming and {AP}-user association,'' \emph{IEEE
  Wireless Commun. Lett.}, vol.~11, no.~2, pp. 298--302, Feb. 2022.

\bibitem{9947028}
G.~Liu, H.~Deng, X.~Qian, W.~Zhang, and H.~Dong, ``Joint pilot and data power
  control for cell-free massive {MIMO} {IoT} systems,'' \emph{IEEE Sens. J.},
  vol.~22, no.~24, pp. 24\,647--24\,657, Dec. 2022.

\bibitem{9915296}
R.~Jia, K.~Xu, X.~Xia, Z.~Shen, W.~Xie, and N.~Sha, ``Time-sequential
  cooperative localization for moving sensor in millimeter-wave cell-free
  massive {MIMO} system,'' \emph{IEEE Sens. J.}, vol.~22, no.~22, pp.
  22\,008--22\,019, Nov. 2022.

\bibitem{rappaport2013millimeter}
T.~S. Rappaport, S.~Sun, R.~Mayzus, H.~Zhao, Y.~Azar, K.~Wang, G.~N. Wong,
  J.~K. Schulz, M.~Samimi, and F.~Gutierrez, ``Millimeter wave mobile
  communications for {5G} cellular: It will work!'' \emph{IEEE Access}, vol.~1,
  pp. 335--349, May. 2013.

\bibitem{8464682}
S.~Sun, T.~S. Rappaport, M.~Shafi, and H.~Tataria, ``Analytical framework of
  hybrid beamforming in multi-cell millimeter-wave systems,'' \emph{IEEE Trans.
  Wirel. Commun.}, vol.~17, no.~11, pp. 7528--7543, Nov. 2018.

\bibitem{heath2016overview}
R.~W. Heath, N.~Gonzalez-Prelcic, S.~Rangan, W.~Roh, and A.~M. Sayeed, ``An
  overview of signal processing techniques for millimeter wave {MIMO}
  systems,'' \emph{{IEEE} J. Sel. Topics Signal Process.}, vol.~10, no.~3, pp.
  436--453, Apr. 2016.

\bibitem{ni2017near}
W.~Ni, X.~Dong, and W.-S. Lu, ``Near-optimal hybrid processing for massive
  {MIMO} systems via matrix decomposition,'' \emph{IEEE Trans. Signal
  Process.}, vol.~65, no.~15, pp. 3922--3933, Aug. 2017.

\bibitem{sabharwal2014band}
A.~Sabharwal, P.~Schniter, D.~Guo, D.~W. Bliss, S.~Rangarajan, and R.~Wichman,
  ``In-band full-duplex wireless: Challenges and opportunities,'' \emph{{IEEE}
  J. Sel. Areas Commun.}, vol.~32, no.~9, pp. 1637--1652, Sept. 2014.

\bibitem{nguyen2020spectral}
H.~V. Nguyen, V.-D. Nguyen, O.~A. Dobre, S.~K. Sharma, S.~Chatzinotas,
  B.~Ottersten, and O.-S. Shin, ``On the spectral and energy efficiencies of
  full-duplex cell-free massive {MIMO},'' \emph{IEEE J. Sel. Areas Commun.},
  vol.~38, no.~8, pp. 1698--1718, Aug. 2020.

\bibitem{wang2019performance}
D.~Wang, M.~Wang, P.~Zhu, J.~Li, J.~Wang, and X.~You, ``Performance of
  network-assisted full-duplex for cell-free massive {MIMO},'' \emph{{IEEE}
  Trans. Commun.}, vol.~68, no.~3, pp. 1464--1478, Mar. 2020.

\bibitem{mohammadi2023network}
M.~Mohammadi, T.~T. Vu, H.~Q. Ngo, and M.~Matthaiou, ``Network-assisted
  full-duplex cell-free massive {MIMO}: {Spectral} and energy efficiencies,''
  \emph{{IEEE} J. Sel. Areas Commun.}, Jun. 2023.

\bibitem{10048919}
S.~Fukue, G.~T. Freitas~de Abreu, and K.~Ishibashi, ``Network-assisted
  full-duplex millimeter-wave cell-free massive {MIMO} with localization-aided
  inter-user channel estimation,'' in \emph{2023 International Conference on
  Information Networking (ICOIN)}, Bangkok, Thailand, Jan. 2023, pp. 13--18.

\bibitem{xia2020joint}
X.~Xia, P.~Zhu, J.~Li, D.~Wang, Y.~Xin, and X.~You, ``Joint sparse beamforming
  and power control for a large-scale {DAS} with network-assisted full
  duplex,'' \emph{IEEE Trans. Veh. Technol.}, vol.~69, no.~7, pp. 7569--7582,
  Jul. 2020.

\bibitem{li2020network}
J.~Li, Q.~Lv, P.~Zhu, D.~Wang, J.~Wang, and X.~You, ``Network-assisted
  full-duplex distributed massive {MIMO} systems with beamforming training
  based {CSI} estimation,'' \emph{{IEEE} Trans. Wireless Commun.}, vol.~20,
  no.~4, pp. 2190--2204, 2020.

\bibitem{li2023network}
\BIBentryALTinterwordspacing
J.~Li, Q.~Fan, Y.~Zhang, P.~Zhu, D.~Wang, H.~Wu, and X.~You, ``Network-assisted
  full-duplex cell-free {mmWave} massive {MIMO} systems with {DAC} quantization
  and fronthaul compression,'' \emph{accepted by China Communications}, 2023.
  [Online]. Available: \url{https://arxiv.org/abs/2302.05571}
\BIBentrySTDinterwordspacing

\bibitem{fan2022maddpg}
Q.~Fan, Y.~Zhang, Z.~Wang, J.~Li, P.~Zhu, and D.~Wang, ``{MADDPG}-based power
  allocation algorithm for network-assisted full-duplex cell-free {MmWave}
  massive {MIMO} systems with {DAC} quantization,'' in \emph{14th International
  Conference on Wireless Communications and Signal Processing (WCSP)}, Nanjing,
  China, Nov. 2022, pp. 556--561.

\bibitem{zhang2020hybrid}
Y.~Zhang, D.~Wang, Y.~Huo, X.~Dong, and X.~You, ``Hybrid beamforming design for
  {mmWave OFDM} distributed antenna systems,'' \emph{Sci. China Inf. Sci.},
  vol.~63, no.~9, pp. 23\,011--230\,112, Jul. 2020.

\bibitem{1597555}
M.~Biguesh and A.~Gershman, ``Training-based {MIMO} channel estimation: a study
  of estimator tradeoffs and optimal training signals,'' \emph{IEEE Trans.
  Signal Process.}, vol.~54, no.~3, pp. 884--893, Mar. 2006.

\bibitem{hong2021effect}
S.-E. Hong, ``On the effect of shadowing correlation and pilot assignment on
  hybrid precoding performance for cell-free mmwave massive {MIMO} {UDN}
  system,'' \emph{ICT Express}, vol.~7, no.~1, pp. 60--70, Mar. 2021.

\bibitem{adhikary2014joint}
A.~Adhikary, E.~Al~Safadi, M.~K. Samimi, R.~Wang, G.~Caire, T.~S. Rappaport,
  and A.~F. Molisch, ``Joint spatial division and multiplexing for {mm-wave}
  channels,'' \emph{IEEE J. Sel. Areas Commun.}, vol.~32, no.~6, pp.
  1239--1255, 2014.

\bibitem{park2018spatial}
S.~Park and R.~W. Heath, ``Spatial channel covariance estimation for the hybrid
  {MIMO} architecture: A compressive sensing-based approach,'' \emph{IEEE
  Trans. Wirel. Commun.}, vol.~17, no.~12, pp. 8047--8062, 2018.

\bibitem{kim2021performance}
I.-s. Kim and J.~Choi, ``Performance of cell-free {MmWave} massive {MIMO}
  systems with fronthaul compression and {DAC} quantization,'' in \emph{2021
  IEEE Wireless Communications and Networking Conference Workshops (WCNCW)},
  Nanjing, China, Mar. 2021, pp. 1--6.

\bibitem{5898372}
J.~Jose, A.~Ashikhmin, T.~L. Marzetta, and S.~Vishwanath, ``Pilot contamination
  and precoding in multi-cell {TDD} systems,'' \emph{IEEE Trans. Wirel.
  Commun.}, vol.~10, no.~8, pp. 2640--2651, Aug. 2011.

\bibitem{8247283}
C.~Pan, H.~Mehrpouyan, Y.~Liu, M.~Elkashlan, and N.~Arumugam, ``Joint pilot
  allocation and robust transmission design for ultra-dense user-centric {TDD}
  {C-RAN} with imperfect {CSI},'' \emph{IEEE Trans. Wirel. Commun.}, vol.~17,
  no.~3, pp. 2038--2053, Mar. 2018.

\bibitem{lei2021maddpg}
W.~Lei, H.~Wen, J.~Wu, and W.~Hou, ``{MADDPG}-based security situational
  awareness for smart grid with intelligent edge,'' \emph{Applied Sciences},
  vol.~11, no.~7, p. 3101, Mar. 2021.

\bibitem{fredj2022distributed}
F.~Fredj, Y.~Al-Eryani, S.~Maghsudi, M.~Akrout, and E.~Hossain, ``Distributed
  beamforming techniques for cell-free wireless networks using deep
  reinforcement learning,'' \emph{IEEE Trans. Cogn. Commun. Netw.}, vol.~8,
  no.~2, pp. 1186--1201, Jun. 2022.

\bibitem{lowe2017multi}
R.~Lowe, Y.~Wu, A.~Tamar, J.~Harb, P.~Abbeel, and I.~Mordatch, ``Multi-agent
  actor-critic for mixed cooperative-competitive environments,'' in
  \emph{Proceedings of the 31st International Conference on Neural Information
  Processing Systems}.\hskip 1em plus 0.5em minus 0.4em\relax Red Hook, NY,
  USA: Curran Associates Inc., Dec. 2017, pp. 6382--–6393.

\bibitem{ackermann2019reducing}
\BIBentryALTinterwordspacing
J.~J. Ackermann, V.~Gabler, T.~Osa, and M.~Sugiyama, ``Reducing overestimation
  bias in multi-agent domains using double centralized critics,''
  \emph{accepted for the Deep RL Workshop at NeurIPS 2019}, 2019. [Online].
  Available: \url{https://api.semanticscholar.org/CorpusID:203642167}
\BIBentrySTDinterwordspacing

\bibitem{zhao2022multi}
N.~Zhao, Z.~Ye, Y.~Pei, Y.-C. Liang, and D.~Niyato, ``Multi-agent deep
  reinforcement learning for task offloading in {UAV}-assisted mobile edge
  computing,'' \emph{IEEE Trans. Wirel. Commun.}, vol.~21, no.~9, pp.
  6949--6960, Sept. 2022.

\bibitem{fozooni2018hybrid}
M.~Fozooni, H.~Q. Ngo, M.~Matthaiou, S.~Jin, and G.~C. Alexandropoulos,
  ``Hybrid processing design for multipair massive {MIMO} relaying with channel
  spatial correlation,'' \emph{IEEE Trans. Commun.}, vol.~67, no.~1, pp.
  107--123, Jan. 2019.

\bibitem{boyd2004convex}
S.~Boyd, S.~P. Boyd, and L.~Vandenberghe, \emph{Convex optimization}.\hskip 1em
  plus 0.5em minus 0.4em\relax Cambridge university press, 2004.

\bibitem{van1993approximation}
\BIBentryALTinterwordspacing
C.~F. Van~Loan and N.~Pitsianis, \emph{Approximation with Kronecker
  Products}.\hskip 1em plus 0.5em minus 0.4em\relax Dordrecht: Springer
  Netherlands, 1993, pp. 293--314. [Online]. Available:
  \url{https://doi.org/10.1007/978-94-015-8196-7_17}
\BIBentrySTDinterwordspacing

\end{thebibliography}

	
	\ifCLASSOPTIONcaptionsoff
	\newpage
	\fi
	
\end{document}